\documentclass[preprint]{elsarticle}

\usepackage{bm,amssymb,amsmath,graphicx,amsthm,wrapfig,caption,subcaption,epstopdf,esint,cancel,xcolor,soul,cancel,ulem}
\usepackage{tikz}
\usetikzlibrary{calc,decorations.markings}

\usepackage[DIV=15]{typearea}

\date{}

\definecolor{curve1}{RGB}{32, 125, 191}
\definecolor{curve2}{RGB}{141, 80, 155}
\definecolor{greenCurve}{RGB}{0, 150, 0}

\newtheorem{remark}{Remark}
\newtheorem{lemma}{Lemma}

\definecolor{curve1}{RGB}{32, 125, 191}
\definecolor{curve2}{RGB}{141, 80, 155}
\newcommand{\Energy}{\mathcal{E}}
\newcommand{\rhot}{\rho}

\definecolor{amber}{rgb}{1.0, 0.49, 0.0}

\begin{document}
\begin{frontmatter}
\title{Ginzburg-Landau model of a Stiffnessometer - a superconducting stiffness meter device}
\author[Math]{Nir Gavish\corref{correspondingauthor}}
\ead{ngavish@technion.ac.il}
\ead[url]{http://ngavish.net.technion.ac.il}
\cortext[correspondingauthor]{Corresponding author}
\author[Physics]{Oded Kenneth}
\ead{kenneth@ph.technion.ac.il}
\ead[url]{https://www.tau.ac.il/~quantum/Kenneth/Kenneth.html}
\author[Physics]{Amit Keren}
\ead{keren@physics.technion.ac.il}
 \ead[url]{http://phsites.technion.ac.il/keren/}
 \address[Math]{Faculty of Mathematics, Technion - Israel Institute of Technology, Haifa, 32000, Israel \fnref{label3}}
\address[Physics]{Faculty of Physics, Technion - Israel Institute of Technology, Haifa, 32000, Israel }

\begin{abstract}
We study the Ginzburg-Landau equations of super-conductivity describing the experimental setup of a Stiffnessometer device.  In particular, we consider the nonlinear regime which reveals the impact of the superconductive critical current on the Stiffnessometer signal.  As expected, we find that at high flux regimes, superconductivity is destroyed in parts of the superconductive regime.  Surprisingly, however, we find that the superconductivity does not gradually decay to zero as flux increases, but rather the branch of solutions undergoes branch folding.   We use asymptotic analysis to characterize the solutions at the numerous parameter regimes in which they exist.  An immediate application of the work is an extension of the regime in which experimental measurements of the Stiffnessometer device can be interpreted.  
\end{abstract}
\end{frontmatter}

\section{Introduction}
Superconductors (SC) are conducting materials that at temperatures lower than some critical value $T_c$ develop two properties:  (I) They loose their electrical resistance, and current can flow in them forever without any voltage, providing that the current density is smaller than a critical value $j_c$. For example, a low enough current in a SC ring will never decay. (II) Small enough magnetic field penetrate a SC only to finite penetration depth $\lambda$. This is known as the Meissner effect. 

High-temperature superconductors are particularly exciting since they operate at temperatures reachable by either liquid nitrogen or electrically powered refrigerators. Their entrance into the consumer market was delayed only by the need to manufacture reliable and flexible wires. Recent development of multi-layered high temperature superconducting based tapes have lead to the manufacturing of several application. Small size orthopedic MRI instruments~\cite{parkinson2017compact}, and bucket size portable 10~T magnets for production lines and laboratories are now available. The next generation of maglav trains will operate with SC~\cite{song2020design}. There are large scale experiment to deliver power and to produce fault current limiters on a city scale based on SC~\cite{yadav2014review}. As consumers confidence in the durability of the wires will grow, so will their application. Consequently, there is a global effort to find and characterize new and better SC. The three important parameters to improve are: $T_c$, $\lambda$, and $j_c$. Multiple methods exist to measure these parameters. They are based on the application of a magnetic field and, more importantly, running current through the SC sample using leads connected to an external source. Having such leads defies the assumptions of thermodynamics and complicates the analysis.  Moreover, measuring $\lambda$ and $j_c$ require different experimental set ups.  

Recently, a new device aka the Stiffnessometer was developed to measure all three parameters of a SC at once, without an external current source connected to the SC~\cite{kapon2017stiffnessometer,kapon2019phase}.  The setup of the measurement consists of a very long coil piercing a SC sample with cylindrical symmetry, see Fig.~\ref{Setup}.  Due to the current in the coil, the SC generates it's own magnetic induction, and vector potential.  The SC parameters can be extracted from proper measurements of the fluxes from both coil and SC cylinder over a range of currents driven through the coil.  Since no external current is applied to the SC, and the SC does not experience external magnetic fields, the measurement is done in thermodynamic conditions, without the interference of vortices anywhere but in the center of the hollow cylinder, and without the complications of sample edges or sample shape.  For more details, we refer the reader to~\cite{kapon2017stiffnessometer,kapon2019phase}.  For completeness, a brief review of the Stiffnessometer measurement device setup and the underlying theory is provided Section~\ref{sec:Stiffnessometer}.  

The Stiffnessometer was successfully applied to study of high temperature superconductors [1,2]. It was demonstrated that the  Stiffnessometer can measure SC properties at relatively high temperature that are closer to the critical temperature~$T_c$ than any other experimental technique. This new detection window, especially close to $T_c$, gives rise to surprising results concerning the behavior of anisotropic SC, like the high $T_c$ materials, close to a point of phase transition. However, the current theory used to interpret the Stiffnessometer measurements is limited to a small flux in the coil corresponding to SC far from the critical current. 

In this work, we study the Ginzburg-Landau equations of super-conductivity in the experimental conditions of the Stiffnessometer.  In particular, we consider the nonlinear regime which reveals the impact of the SC critical current on the Stiffnessometer signal.  We determine quantitatively the relation between the measured vector potential of the hollow SC cylinder to the applied flux in the coil, for different SC parameters.

\begin{figure}[tbph]
\center
	\includegraphics[trim=0cm 0cm 0cm 0cm, clip=true,width=0.36\columnwidth]{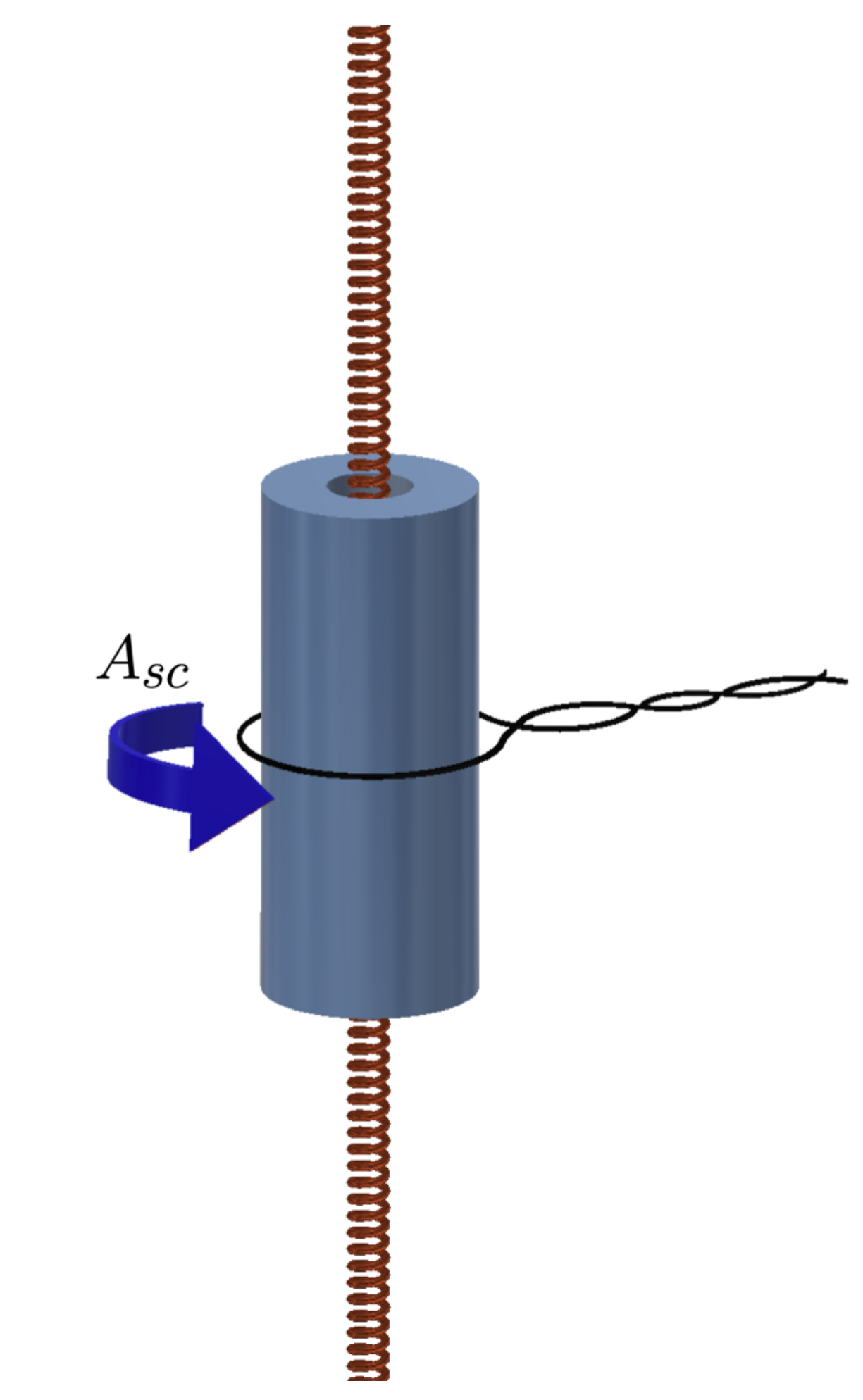}
	\caption{Illustration of the Stiffnessometer principle of operation. An ideally infinitely long coil pierces a hollow superconducting cylinder. The cylinder's height is much larger than it's radial dimensions. Applied current is running through the coil generating flux in its center. The magnetic field outside of the coil is zero. A pickup loop surrounds both coil and cylinder and measures the total flux $\Phi$ from both. By subtracting the coil flux measured above the superconductor $T_c$, from the total flux measure below $T_c$, the flux generated by the superconductor $\Phi_{sc}$ can be determined. $\Phi_{sc}$ is proportional to the vector potential~$A_{\rm sc}$ of the cylinder at the pickup loop location (see text). }
	\label{Setup}
\end{figure}
\subsection{Paper outline and summary of results}
The paper is organized as follows: In Section~\ref{sec:Stiffnessometer} we describe the experimental setup and derive the essential equations required to understand the phenomena of superconductivity. In Section~\ref{sec:model} we present the derivation of the Ginzburg-Landau system that describes the Stiffnessometer setup considered in this work.  Particularly, we show that the Ginzburg-Landau system can be expressed in terms of numerous relevant physical quantities in various domains and involves several parameters on different scales.  A key challenge addressed in this section is the choice of the quantities of study.  In Section~\ref{sec:analysis} we consider analytic properties of the solutions of the corresponding Ginzburg-Landau system, e.g., we prove they are monotone.  Section~\ref{sec:superconductive} presents a study of the case for which the whole cylinder region is in a superconductive state.  We approximate the solutions of the Ginzburg-Landau system in this case.  Particularly, we show that the vector potential is well approximated by an explicit function expressed in terms of Bessel functions, and that the superconducting order parameter has a double boundary layer near the inner rim of the superconducting hollow cylinder.  These approximations are used to quantify the flux regime in which the whole cylinder region is in a superconductive state, aka the low flux regime.  In Section~\ref{sec:highCurrentRegime} we consider the behavior of the system at higher fluxes.  Particularly, in Section~\ref{sec:partialSC}
we consider the case of partial superconductivity in which, roughly speaking, not all the superconducting cylindrical region is in a full superconductive state.  We present an approximation of the solution of the nonlinear equation describing the vector potential in this case, and quantify the flux regime in which these solutions persist. In Section~\ref{sec:SC_partially_destroyed} we consider yet higher flux regimes in which superconductivity is destroyed in part of the cylindrical region. We show that in this case, the study of the order parameter equation gives rise to a nonlinear turning point problem, and that near the transition point the order parameter is well approximated by a scaled Hastings-McLeod solution. This result, however, does not reveal the location of the turning point.  Using variational methods, we approximate its location, and use this result to approximate the vector potential in this case. The equation for the vector potential is nonlinear in this case. Nevertheless, we show 
that the vector potential can be approximated by the solution of the low flux regime linear equation applied to a superconducting cylinder region which has an effective inner radius related to the location of the turning point. The study of the system in the high flux regime shows that superconductivity decreases with the normalized flux~$J$, but not up to a point in which superconductivity is completely destroyed. In Section~\ref{sec:weakSC}, we study the case of weak superconductivity in which superconductivity is nearly completely destroyed. In contrast to the previous sections, we do not make any assumptions on the magnitude of $J$. Rather, we aim to obtain from the analysis an approximation of $J$ at which superconductivity is destroyed. Surprisingly, we find that superconductivity is destroyed in a regime of fluxes which are significantly smaller than the fluxes in the `high flux' regime. Moreover, we find that, in this regime of fluxes, the superconducting state order parameter increases with the flux. This implies that for a range of flux values, the system has multiple non-trivial solutions. To understand whether these solutions belong to different solution branches, or belong to one solution branch that undergoes branch folding, we conduct in Section~\ref{sec:numericalContinuation} a numerical continuation study.   Particularly, we show that the branch of nontrivial solutions described in Sections~\ref{sec:superconductive} and~\ref{sec:highCurrentRegime} undergoes a branch folding in the high flux regime, where the section of the branch after the branch fold is described in Section~\ref{sec:weakSC}.  Numerical details are provided in Section~\ref{sec:numerics}. Concluding remarks are presented in Section~\ref{sec:ConcludingRemarks}.

\section{Brief review of the Stiffnessometer device}\label{sec:Stiffnessometer}
We now provide a brief review of the Stiffnessometer device, and the related theory, in aim of providing proper context to this work.  We refer the reader to~\cite{kapon2017stiffnessometer,kapon2019phase} for additional details.  In particular, we present a derivation of the underlying equations that relate the measured quantities to the relevant SC parameters using London's equation.  In the next section, we present a variational derivation based on the Ginzburg-Landau free energy which, in proper regimes, gives rise to the same equations.  The reasons for this redundancy are two-fold: I) Physical quantities parameters relevant to the Stiffnessometer device arise more naturally in the derivation based on London's equation.  II) It allows readers who would like to focus on other aspects of the problem to skip this section.  

The Stiffnessometer device setup is based on a very long coil, approximated here by an infinite coil, piercing a SC sample with cylindrical symmetry. Here we consider a very tall hollow cylinder as in Fig.~\ref{Setup}. Both coil and SC are first cooled to a temperature below $T_c$ and then a current is ramped in the coil generating time dependent flux only in its interior. According to Faraday's law 
\[\nabla  \times {\bf{E}} =  -\frac{1}{c} \frac{{\partial {\bf{B}}}}{{\partial t}}\]
where $\bf{B}$ is the magnetic induction, $c$ is the speed of light, an $\bf{E}$ is the electric field which develops in the SC cylinder until the magnetic induction reaches it's final value. Defining a vector potential via 
\begin{equation}
{\bf{B}} = \nabla  \times {\bf{A}}
\label{DefOfA}
\end{equation}
ensures 
\[{\bf{E}} =  -\frac{1}{c} \frac{{\partial {\bf{A}}}}{{\partial t}} + \nabla U\] where $U$ is an arbitrary function. $\bf{A}$ is not determined uniquely and the gradient of any function could be added to it; a property known as a gauge freedom. We define further 
\[ U = \frac{\hbar }{e^*} \frac{{\partial \phi }}{{\partial t}}\]
where $\hbar$ is the Planck constant and ${e^*}$ is the carriers charge. We assume, with a grain of salt, that a SC can be described by friction free motion of the charge carriers. We will return to this assumption shortly, but for now, its consequence is that the current density in the SC cylinder ${\bf{j}}$ is given by 
\begin{equation}
{\bf{j}} = n{e^*}{\bf{V}} = \frac{{{ne^{*2}}}}{{{m^*}}}\int\limits_0^t {{\bf{E}}dt}  =  - {\rho _s}({\bf{A}} - \frac{{\hbar c}}{{{e^*}}}\nabla \phi )
\label{London}
\end{equation}
where 
\[{\rho _s} = \frac{{n{{e^*}^2}}}{{m^*c}}\] 
is called stiffness, ${m^*}$,  $n$, and $\bf{V}$ are the carriers mass, density, and velocity, respectively. This is known as the London equation in its gauge invariant form. An embedded assumption in our derivation is that the carriers move around the cylinder in circles experiencing a position independent electric field, hence the partial derivative and full integration with respect to time cancel each-other. 

It is important to mention that the London equation has broader validity than the derivation presented here. It is valid for all SC shapes, and it predicts the Meissner effect even if the field is turned on before the sample is cooled, and current flow in the SC is dissipative. To obtain the Meissner effect, one takes the rotor of Ampere's law in a steady state
\begin{equation}
\nabla  \times {\bf{B}} = \frac{{4\pi }}{c}{\bf{j}}
\label{Ampere}
\end{equation}
and uses Eq. \ref{London} to generate the partial differential equation 
\[\nabla  \times \nabla  \times {\bf{B}} =-\Delta {\bf{B}}=  - \frac{{4\pi }}{c}{\rho _s}{\bf{B}}\] 
who's solution is a magnetic induction decaying into the sample exponentially with
\[\frac{1}{{{\lambda ^2}}} = \frac{{4\pi }}{c}{\rho _s}.\]
The Meissner effect and zero resistance are intermittently related since to expel $\bf{B}$ out of a SC, current must run in it indefinitely without dissipation, and produce an opposing field to the applied one.

Before the invention of the Stiffnessometer all measurements of stiffness where done via the penetration depth $\lambda$, and not by the original definition of Eq.~\ref{London}. Finally, due to the current $\bf{j}$, the SC generates it's own magnetic induction, and vector potential. The real $\bf{j}$ in the SC is due to both the applied and self vector potentials as will be discussed below.

In the experiment, a pickup loop circles both the coil and SC cylinder, as depicted in Fig.~\ref{Setup}. The pickup loop is connected to a superconducting quantum interference device (SQUID), and measures the flux $\Phi $ from both coil and SC cylinder through the loop. By subtracting measurement above $T_c$ from a measurement below $T_c$ the flux contribution from the cylinder $\Phi_{sc} $ can be determined. This flux is given by 
\[{\Phi _{sc}} = 2\pi {R_{pl}}{A_{sc}(R_{pl})}\]
where $R_{pl}$ is the pickup loop radius, and $A_{sc}(R_{pl})$ is the vector potential generated by the SC cylinder only in the pickup loop position.

The relation between the applied flux and the measured SC flux $\Phi _{sc}$ is determined by the stiffness, hence the instrument's name. Moreover, it is expected that for high enough applied flux, the current density in the cylinder will be so large that it will cross $j_c$ and SC will be destroyed allowing for the measurement of $j_c$ simultaneously with the stiffness. 

The superconducting carrier density $n$ in a SC is temperature dependent and undergoes a phase transition. It starts from zero at $T_c$ and increases as the temperature is lowered. In addition, SC are dissipation free, hence can be described by a free energy. Consequently, Ginzburg  and Landau (GL) invented a free energy which treat $n$ as an order parameter, on one hand, and produce the London equation on the other hand. The variable of this theory is an order parameter $\Psi(\bf{r})$, which is a complex, space dependent function. 
${\left| \Psi  \right|^2}$ is proportional to $n$. Its phase is the same $\phi(\bf{r})$ as in Eq.~\ref{London}. Treating $\phi$ as a phase allow it to be a multi valued function. Magnetic induction enter the theory either via a vector potential $\bf{A}$ or as $\nabla \times \bf{A}$ to account for the energy associated with $\bf{B}$ at temperatures above $T_c$.

\section{Model derivation}\label{sec:model}
The Ginzburg Landau free energy reads as
\[
\Energy=\int\left(\frac{|\nabla\times {\bf A}|^2}{8\pi}+\frac1{2m^*}\left|\left(\frac{\hbar}{i}\nabla-\frac{e^*}{c}{\bf A}\right)\Psi\right|^2+\alpha|\Psi|^2+\frac{\beta}2|\Psi|^4\right)d{\bf x},
\]
where the variables are the same variable as used in the derivation of Section~\ref{sec:Stiffnessometer}, and are defined for completeness:~${\bf A}$ is the vector potential,~$\Psi$ is the superconducting state order parameter,~$m^*$ is an effective pair mass,~$\hbar$ is the Planck constant,~$e^*$ is the charge of the SC carriers,~$c$ is the speed of light, and~$\alpha<0<\beta$ are the phenomenological Ginzburg-Landau coefficients.  The order parameter ${\left| \Psi  \right|^2}$ is proportional to the superconducting carrier density.  

The spatial variables are normalized by a reference radius~$R_{\rm pl}$, which in this case is taken to be the radius of the pickup loop,
\[
\tilde {\bf x}=\frac{\bf x}{R_{\rm pl}},
\]
the vector potential is normalized according to
\[
\tilde{\bf A}=\frac{{\bf A}}{A_0},\quad A_0=\frac{\Phi_0}{2\pi R_{\rm pl}},\quad \Phi_0=\frac{2\pi \hbar c}{e^*},
\]
where~$\Phi_0$ is the flux quanta,
and the superconducting state order parameter takes the polar form~$\Psi=\psi e^{i\phi}$ 
\[
\tilde \psi=\frac{\psi}{\psi_\infty},\quad \psi_\infty=\sqrt{-\frac{\alpha}{\beta}}.
\]
$\left|\tilde{\psi} \right|=1$ accounts for maximal  superconducting carrier density or a full superconductive state and~$\left| \tilde \psi  \right|=0$ accounts for a non-superconductive state. In what follows, we consider only the non-dimensional variables and omit the tildes.  
The non-dimensional Ginzburg Landau free energy reads as
\[
\Energy=\frac12\int \left(\lambda^2\varepsilon^2 |\nabla\times\vec A|^2+\varepsilon^2|\nabla \psi|^2+\varepsilon^2 \psi^2|\vec A-\nabla\phi|^2+\frac12\psi^4-\psi^2\right)\,d{\bf x}.
\]
where~$\varepsilon$ is the normalized coherence length and~$\lambda$ is the normalized penetration depth,
\[
\varepsilon^2=-\frac1{R_{\rm pl}^2}\frac{\hbar^2}{2m^* \alpha},\quad \lambda^2=-\frac1{R_{\rm pl}^2}\frac{\beta}{\alpha}\frac{m^* c^2}{4\pi (e^*)^2}.
\]

Let us consider a system of a superconductive hollow cylinder in the region~$0<r_{\rm in}<r<r_{\rm out}$ and an infinitely long coil of infinitesimal diameter on the z axis such that its own vector potential is
\[
\vec A_{\rm coil}=\frac{J}{r}\hat \theta,
\]
where~$J=\Phi/\Phi_0$, and~$\Phi$ is the flux inside the coil.
Due to the symmetry of the system, the vector potential is tangential and depends only on the radius
\[
\vec{A}=A(r)\,\hat\theta=\left[A_{\rm sc}(r)+A_{\rm coil}\right]\hat\theta,
\]
where~$A(r)$ and~$A_{\rm sc}$ are the total and superconductor tangential components of the vector potential, respectively. The order parameter satisfies in the cylinder region
\begin{equation}
\psi=\psi(r)\ge 0,\quad \nabla\phi=\frac{m}{r}\hat\theta,\quad m\in Z,\quad r_{\rm in}<r<r_{\rm out},
\label{eq:parameterslimits}
\end{equation}
and~$\psi\equiv0$ outside the cylinder region.  Note that~$\psi\ge 0$ since it is the absolute value of the superconducting state order parameter, and the above form of $\nabla\phi$ is a gauge choice, namely, the Coulomb gauge.

The Ginzburg-Landau free energy takes the form
\begin{equation}\label{eq:energyAsc}
\Energy=\pi\int_{r=0}^\infty \left\{\lambda^2\varepsilon^2 \left(A^\prime_{\rm sc}(r)+\frac{A_{\rm sc}}{r}\right)^2+\varepsilon^2\psi_r^2+\varepsilon^2\left(A_{\rm sc}+\frac{J-m}{r}\right)^2\psi^2(r)+
\frac{\psi^4}2-\psi^2\right\}\,rdr.
\end{equation}
The Ginzburg-Landau equation 
\[
\frac{\delta \Energy}{\delta A_{\rm sc}}=0,\quad \frac{\delta \Energy}{\delta \psi}=0,
\]
 reads as
\begin{subequations}\label{eq:system}
\begin{equation}\label{eq:Asc}
A^{\prime\prime}_{\rm sc}(r)+\frac{A^\prime_{\rm sc}}r-\frac{A_{\rm sc}}{r^2}=\frac{1}{\lambda^2}\left(A_{\rm sc}+\frac{J-m}{r}\right)\psi^2(r),\quad A_{\rm sc}(0)=0=A_{\rm sc}(\infty),\qquad r>0,
\end{equation}
and
\begin{equation}\label{eq:g}
\varepsilon^2\left(\psi^{\prime\prime}(r)+\frac{\psi^\prime}r\right)
=\psi^3-\left(1-\varepsilon^2\left(A_{\rm sc}+\frac{J-m}{r}\right)^2\right)\psi,\quad \psi(r)\ge0,\quad \psi^\prime(r_{\rm in})=\psi^\prime(r_{\rm out})=0, \qquad r_{\rm in}<r<r_{\rm out}.
\end{equation}
\end{subequations}
Equation~\eqref{eq:Asc} is Ampere's law, see~\eqref{Ampere}, for the SC magnetic induction in terms of $A_{sc}$ from~\eqref{DefOfA}, and $j$ from~\eqref{London}, which in turn, is determined by the sum of $A_{\rm sc}$ and $A_{\rm coil}$. Outside the cylinder region,~$\psi\equiv0$ and~$A_{\rm sc}$ satisfies the homogenous equation 
\[
A^{\prime\prime}_{\rm sc}+\frac1rA_{\rm sc}^{\prime}-\frac1{r^2}A_{\rm sc}=0,
\]
whose solution is of the general form
\[
A_{\rm sc}^{\rm homogenous}=c_1 r+\frac{c_2}r.
\]
Specifically, for~$0<r<r_{\rm in}$, the boundary condition~$A_{\rm sc}(0)=0$ implies that~$A_{\rm sc}=c_1r$.  Continuity of~$A_{\rm sc}$ at~$r=r_{\rm in}$ further implies that~$c_1=A_{\rm sc}(r_{\rm in})/r_{\rm in}$.
Hence, at the point~$r=r_{\rm in}$, the solution satisfies
\begin{subequations}\label{eq:system_BC_cylinder}
\begin{equation}\label{eq:BCL}
A_{\rm sc}^\prime(r_{\rm in})-\frac{A_{\rm sc}(r_{\rm in})}{r_{\rm in}}=0.
\end{equation}
For~$r>r_{\rm out}$, the boundary condition~$A_{\rm sc}(\infty)=0$ and the continuity of~$A_{\rm sc}$ implies that~$A_{\rm sc}=c_2/r$ where~$c_2=A_{\rm sc}(r_{\rm out})r_{\rm out}$. Hence, at the point~$r=r_{\rm out}$,
\begin{equation}\label{eq:BCR}
A_{\rm sc}^\prime(r_{\rm out})+\frac{A_{\rm sc}(r_{\rm out})}{r_{\rm out}}=0.
\end{equation}
\end{subequations}
Overall, the solution in the whole domain is of the form
\begin{equation}\label{eq:rel_cylinder_domain_to_R}
\begin{cases}
A_{\rm sc}(r_{\rm in})\frac{r}{r_{\rm in}}& r<r_{\rm in},\\
A_{\rm sc}(r)& r_{\rm in}\le r\le r_{\rm out},\\
A_{\rm sc}(r_{\rm out})\frac{r_{\rm out}}{r}& r>r_{\rm out}.
\end{cases}
\end{equation}
The above equations describe a wide range of superconductors with different ratios of~$\lambda$ to~$ \varepsilon$.
In what follows, we will consider the asymptotic regime valid for high temperature superconductors
\[
\varepsilon\ll \lambda\ll1.
\]

Furthermore, without loss of generality, we consider the case~$m=0$. The system~\eqref{eq:system} can be expressed in terms of numerous relevant physical quantities such as~$A$ or~$A_{\rm sc}$.  It also involves several parameters on different scales, e.g.,~$\varepsilon$,~$\lambda$ and~$J$, and can be solved in various domains, e.g., the whole domain or the cylinder region.  A key challenge of this work is to choose which quantity to study.  In what follows, we consider the total vector potential scaled by the normalized flux~$J$,
\begin{equation}\label{eq:AJ}
A_J=\frac1{J}\left(A_{\rm sc}+\frac{J}{r}\right).
\end{equation}
Additionally, we consider a system for both~$A_J$ and~$\psi$ in the cylinder region~$0<r_{\rm in}<r<r_{\rm out}$. 

Substituting~\eqref{eq:AJ} in~\eqref{eq:system} and using the boundary conditions~\eqref{eq:system_BC_cylinder} gives rise to the following system for~$A_J$ and~$\psi$,
\begin{subequations}\label{eq:system_ring_region}
\begin{equation}\label{eq:A}
A_J^{\prime\prime}(r)+\frac{A_J^\prime}r-\frac{A_J}{r^2}=\frac{1}{\lambda^2}A_J(r)\psi^2(r),\qquad A_J^\prime(r_{\rm in})-\frac{A_J}{r_{\rm in}}=-\frac{2}{r_{\rm in}^2},\quad A^\prime_J(r_{\rm out})+\frac{A_J}{r_{\rm out}}=0,
\end{equation}
and
\begin{equation}\label{eq:psi}
\varepsilon^2\left(\psi^{\prime\prime}(r)+\frac{\psi^\prime}r\right)
=\psi^3-\left(1-\varepsilon^2J^2A_J^2(r)\right)\psi,\quad \psi(r)\ge0,\quad \psi^\prime(r_{\rm in})=\psi^\prime(r_{\rm out})=0.
\end{equation}
\end{subequations}
The solution of~\eqref{eq:system_ring_region} in the hollow cylinder domain is related to~$A_{\rm sc}$ in the whole domain by~\eqref{eq:rel_cylinder_domain_to_R} and~\eqref{eq:AJ}.
In the subsequent sections, we will show that the above choices, and particularly the choice of~$A_J$, open the way to analysis of the equation.
\section{Analysis}\label{sec:analysis}
Let us consider solutions~$(A_J,\psi)$ of the system~\eqref{eq:system_ring_region}.  In this section, we present the analytic properties of these solutions. Particularly, we prove their monotonicity and provide bounds for their values.

\subsection{Analysis of equation~\eqref{eq:A} for~$A_J$}
\begin{lemma}\label{lem:extremumAJ}
Let~$A_J$ be a solution of~\eqref{eq:A} for a given function~$\psi$, and let~$r_c$ be a critical point of~$A_J$ in~$(r_{\rm in},r_{\rm out})$.   If~$A_J>0$ then~$r_c$ is a strict local minimum point of~$A_J$, and if~$A_J<0$ then~$r_c$ is a strict local maximum point of~$A_J$.  
\end{lemma}
\begin{proof}
The point~$r_c$ is a critical point of~$A_J$, hence~$A_J^\prime(r_c)=0$.  Therefore, by~\eqref{eq:A}, 
\[
A_J^{\prime\prime}(r_c)=\left[\frac1{r_c^2}+\frac{\psi^2(r_c)}{\lambda^2}\right]A_J(r_c).
\]
Hence,~${\rm sign}[A_J^{\prime\prime}(r_c)]={\rm sign}[A_J(r_c)]$.
\end{proof}
\begin{lemma}\label{lem:AJ}
Let~$A_J$ be a solution of~\eqref{eq:A}.  Then, for~$r_{\rm in}< r< r_{\rm out}$,
\begin{equation}\label{eq:lem_AJ_result}
0< A_J<\frac{2}{r_{\rm in}},\quad A_J^\prime < 0.
\end{equation} 
\end{lemma}
\begin{proof}
We first prove that~\eqref{eq:lem_AJ_result} holds at~$r=r_{\rm in}$.  Let us assume in negation that~$A_J(r_{\rm in})\le0$.  Then, the boundary condition at~$r=r_{\rm in}$ implies that~$A^\prime_J(r_{\rm in})<0$.  
Hence, there exists a surrounding of~$r_{\rm in}$ for which~$A_J(r)<0$ and~$A^\prime_J(r)<0$.  The boundary condition at~$r=r_{\rm out}$ implies that either~$A^\prime(r_{\rm out})>0$ or~$A(r_{\rm out})\ge0$.  In both cases, 
there exists a minimum point of~$A_J(r)$ in~$(r_{\rm in},r_{\rm out})$ at which~$A_J<0$, in negation with Lemma~\ref{lem:extremumAJ}. 
Similarly,~$A_J(r_{\rm in})\ge\frac{2}{r_{\rm in}}$ implies that~$A^\prime_J(r_{\rm in})\ge 0$.  The boundary condition at~$r_{\rm out}$ implies in this case that there exists a local maximum point of~$A_J(r)$ in~$(r_{\rm in},r_{\rm out})$ at which~$A_J>0$, in negation with Lemma~\ref{lem:extremumAJ}.  Therefore,~$0<A_J(r_{\rm in})<\frac{2}{r_{\rm in}}$ and~$A_J^\prime(r_{\rm in})<0$.

 We next prove that~$A_J$ is strictly monotonically decreasing~$(r_{\rm in},r_{\rm out})$.  Let us assume, in negation, that~$A_J$ does not decrease monotonically. In this case, there exists a critical point~$r^*$ of~$A_J$, such that~$A_J^\prime(r)<0$ for~$r_{\rm in}\le r<r^*$.  By Lemma~\ref{lem:extremumAJ},~$A_J(r^*)\ge0$ since~$r^*$ is not a local maximum point of~$A_J$.  Therefore,
\begin{equation}\label{eq:lem2_A_positive_r*}
A_J(r)>0,\quad r_{\rm in}<r<r^*.
\end{equation}
We first rule out the option~$A_J(r^*)=0$.  Indeed, Equation~\eqref{eq:A} is a second order linear ODE with continuous variable coefficients in the cylinder region~$[r_{\rm in},r_{\rm out}]$ (Note that~$r_{\rm in}>0$).  Therefore, the associated initial value problem
\[
A_J^{\prime\prime}(r)+\frac{A_J^\prime}r-\frac{A_J}{r^2}=\frac{1}{\lambda^2}A_J(r)\psi^2(r),\quad A(r^*)=A^\prime(r^*)=0,\qquad r_{\rm in}<r<r^*,
\]
has a unique solution, which in this case is~$A_J\equiv0$.  In contradiction with~\eqref{eq:lem2_A_positive_r*}.
We next rule out the option~$A_J(r^*)>0$.  If~$A_J(r^*)>0$, then by Lemma~\ref{lem:extremumAJ},~$r^*$ is a local minimum point of~$A_J$.  There are again two options: If~$A_J^\prime\ge0$ in~$(r^*,r_{\rm out})$, then~$A_J(r_{\rm out})>0$ and~$A_J^\prime(r_{\rm out})\ge0$ and hence~$A_J$ does not satisfy the boundary condition at~$r=r_{\rm out}$, see, e.g., red dashed curve in Figure~\ref{fig:Lemma2Illustration}.  The second option is that~$A_J$ is not monotonically increasing in~$(r^*,r_{\rm out})$.  In this case, there exists additional extremum points in~$(r^*,r_{\rm out})$, and in particular a local maximum point with~$A_J>0$, in contraction with Lemma~\ref{lem:extremumAJ}.  Therefore,~$A_J$ must decrease monotonically.

We finally prove that~$A_J$ is strictly positive in~$(r_{\rm in},r_{\rm out})$.
Since~$A_J$ is strictly monotonically decreasing in~$(r_{\rm in},r_{\rm out})$, the boundary condition at~$r=r_{\rm out}$ implies that~$A_J(r_{\rm out})\ge 0$, hence~$A_J>0$ in~$(r_{\rm in},r_{\rm out})$, see, e.g., blue solid curve in Figure~\ref{fig:Lemma2Illustration}.
\end{proof}

\begin{figure}[ht!]
\begin{center}
\scalebox{0.5}{\includegraphics{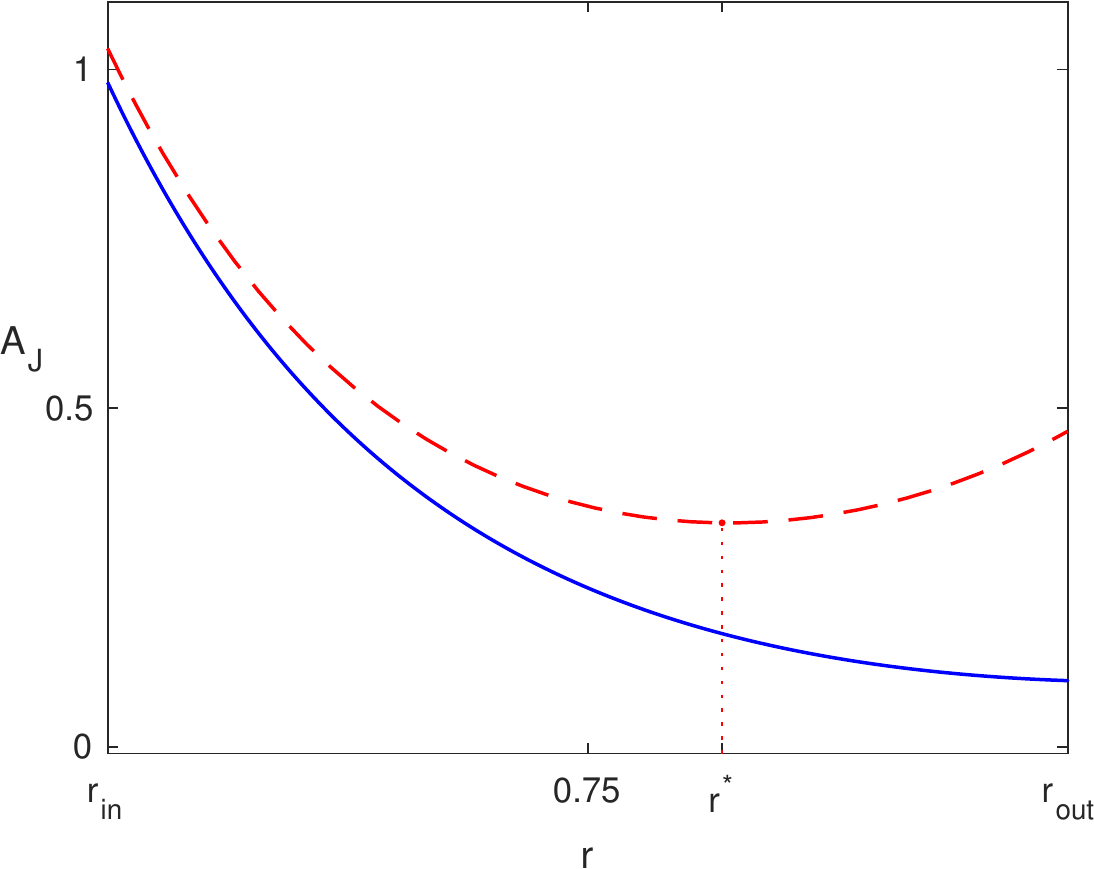}}
\caption{Solution~$A_J$ of boundary value problem~\eqref{eq:A} for~$r_{\rm in}=0.5$,~$r_{\rm out}=1.5$,~$\lambda=0.2$ and~$\psi=\sqrt{1-0.01e^{-r}}$ ({\color{blue} solid}). Also plotted is a solutions~$A^\pm_{\rm IVP}$ of the associated initial value problem~\eqref{eq:A} with the same parameters as above, with the initial condition~$A_{\rm IVP}(r_{\rm in})=c$ and~$A^\prime_{\rm IVP}(r_{\rm in})=c-2/r_{\rm in}^2$ where~$c=A_J(r_{\rm in})+0.1$ ({\color{red} dashes}). }
\label{fig:Lemma2Illustration}
\end{center}

\end{figure}

\subsection{Analysis of equation~\eqref{eq:psi} for~$\psi$}

\begin{lemma}\label{lem:psi_re}
Let~$\psi$ be a solution of~\eqref{eq:psi} for a given function~$A_J(r)$,  let~$r_c$ be a critical point of~$\psi$ in~$[r_{\rm in},r_{\rm out}]$, and define
\begin{equation}\label{eq:f_psi}
f(r)=
\begin{cases}
\sqrt{1-\varepsilon^2J^2A_J^2(r)},&1-\varepsilon^2J^2A_J^2(r)\ge0,\\
0& \mbox{otherwise}.
\end{cases}
\end{equation}
Then, if~$\psi(r_c)>f(r_c)$,~$r_c$ is a local minimum point of~$\psi$ and if~$\psi(r_c)<f(r_c)$,~$r_c$ is a local maximum point of~$\psi$. 
\end{lemma}
\begin{proof}
The point~$r_c$ is a critical point of~$A_J$, hence~$\psi^\prime(r_c)=0$.  
When~$1-\varepsilon^2J^2A_J^2(r_c)\ge0$, by~\eqref{eq:psi} and the definition~\eqref{eq:f_psi},
\[
\psi^{\prime\prime}(r_c)=\frac1{\varepsilon^2}\left[\psi^2(r_c)-f^2(r_c)\right]\psi(r_c).
\]
Otherwise,~$f=0$ and~$\psi^{\prime\prime}(r_c)>0$.
In both cases,~$\mbox{sign}(\psi^{\prime\prime}(r_c))=\mbox{sign}(\psi(r_c)-f(r_c))$.  

Therefore, when~$\psi(r_c)<f(r_c)$,~$\psi^{\prime\prime}(r_c)>0$, and~$r_c$ is a local minimum point of~$\psi$.  Similarly, when~$\psi(r_c)>f(r_c)$,~$\psi^{\prime\prime}(r_c)<0$, and~$r_c$ is a local maximum point of~$\psi$.   
\end{proof}
Note, in particular, that the end points~$r_{\rm in}$ and~$r_{\rm out}$ are critical points of~$\psi$, and that Lemma~\ref{lem:psi_re} applied to these  points.
\begin{figure}[ht!]
\begin{center}
\scalebox{0.5}{\includegraphics{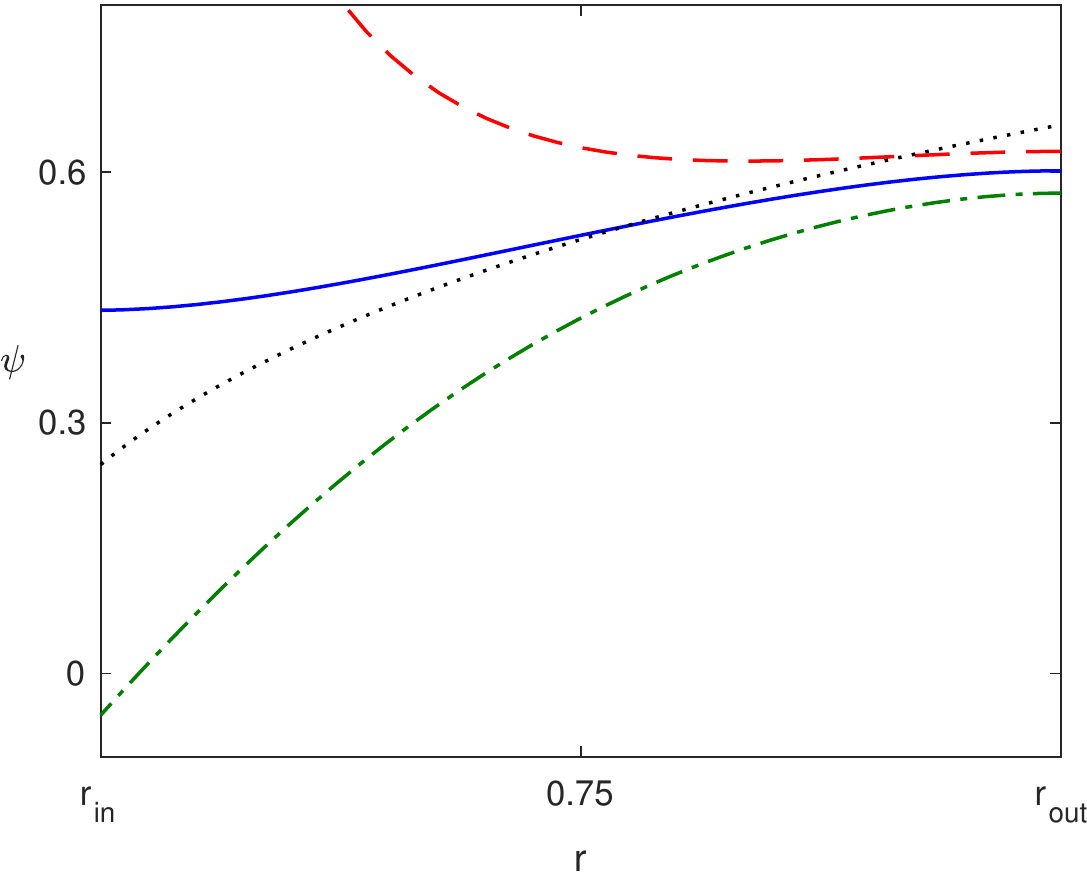}}
\caption{Solution~$\psi$ of boundary value problem~\eqref{eq:psi} for~$r_{\rm in}=0.5$,~$r_{\rm out}=1.5$,~$\varepsilon=0.1$ and~$\varepsilon^2 J^2 A_J^2=1.54e^{-r}$ ({\color{blue} solid}). Also plotted are two solutions~$A^\pm_{\rm IVP}$ of the associated initial value problem~\eqref{eq:psi} with the same parameters as above and with~$\psi^+_{\rm IVP}(r_{\rm in})=\psi(r_{\rm out})+0.025$ ({\color{red} dashes}) and~$\psi^-_{\rm IVP}(r_{\rm in})=\psi(r_{\rm in})-0.025$ ({\color{greenCurve} dash-dots}).  Super-imposed is the curve~$f$ defined by~\eqref{eq:f_psi} (dots).}
\label{fig:Lemma4Illustration}
\end{center}
\end{figure}\begin{lemma}\label{lem:psi}
Let~$\psi$ be a solution of~\eqref{eq:psi} for a given function~$A_J(r)$ which satisfies~\eqref{eq:lem_AJ_result}.  Then,
\begin{equation}\label{eq:lem4res}
0\le f(r_{\rm in})\le \psi \le f(r_{\rm out})\le1,\quad \psi^\prime\ge0,
\end{equation}
where~$f$ is given by~\eqref{eq:f_psi}.
\end{lemma}
\begin{proof}
We first  rule out the possibility that~$r=r_{\rm out}$ is a local minimum point of~$\psi$ from the left.  By~\eqref{eq:psi},~$r=r_{\rm out}$ is a critical point of~$\psi$.  If~$r=r_{\rm out}$ is a local minimum point of~$\psi$ from the left, ,then~$\psi^\prime<0$ in a surrounding~$r<r_{\rm out}$.
 Moreover, the boundary condition~$\psi^\prime(r_{\rm in})=0$ ensures that there exists a point~$r_{\rm in}\le r_1<r_{\rm out}$ such that~$\psi^\prime(r_1)=0$ and~$\psi^\prime(r)<0$ for~$r_1<r<r_{\rm out}$.  
Particularly,
\begin{equation}\label{eq:lem4_psi_rc}
\psi(r_1)>\psi(r_{\rm out})
\end{equation}
The point~$r_1$ is not a local minimum point of~$\psi$ since~$\psi^\prime(r)<0$ for~$r_1<r<r_{\rm out}$.  Hence, by Lemma~\ref{lem:psi_re},
\begin{equation}\label{eq:lem4_conc1}
\psi(r_1)\le f(r_1).  
\end{equation}
By Lemma~\ref{lem:AJ},~$f(r)$ is a monotonically increasing function in~$r$.  Therefore,
\[
\psi(r_{\rm out})<\psi(r_1)\le f(r_1)\le f(r_{\rm out})
\]
where the first two inequalities are due to~\eqref{eq:lem4_psi_rc} and~\eqref{eq:lem4_conc1}.
However, the assumption~$r=r_{\rm out}$ is a local minimum point of~$\psi$ implies, by Lemma~\ref{lem:psi_re},
\[
\psi(r_{\rm out})\ge f(r_{\rm out}).
\]
In contradiction.

We next consider the case~$\psi(r)\equiv\psi(r_{\rm out})$.  Since in this case~$\psi^\prime(r_{\rm out})=\psi^{\prime\prime}(r_{\rm out})=0$, by Lemma~\ref{lem:psi_re},~$\psi(r)\equiv f(r_{\rm out})$.  Lemma~\ref{lem:AJ} implies that~$f$ is monotonically non-decreasing, hence~$f(r_{\rm in})\le f(r_{\rm out})$.  We now rule out the possibility that~$f(r_{\rm in})< f(r_{\rm out})$.  Indeed, in this case~$f(r_{\rm in})<f(r_{\rm out})=\psi(r_{\rm in})$.  Hence, by Lemma~\ref{lem:psi_re},~$r=r_{\rm in}$ is a local minimum point of~$\psi$.  In contradiction.  Therefore,~$\psi(r)\equiv\psi(r_{\rm out})$ is possible only when~$f(r)\equiv f(r_{\rm out})$.  By Lemma~\ref{lem:AJ},  this is the case,~$1-\varepsilon^2J^2A_J^2(r)<0$ in the cylinder region, hence~$f\equiv0$ and~$\psi\equiv0$.  In this case, the solution satisfies conditions~\eqref{eq:lem4res}.

Finally, we consider the case that~$r=r_{\rm out}$ is a local maximum point of~$\psi$ from the left.  By Lemma~\ref{lem:psi_re}, in this case,
\[
\psi(r_{\rm out})\le f(r_{\rm out}).
\]
The boundary condition~$\psi^\prime(r_{\rm in})=0$ ensures that there exists a point~$r_{\rm in}\le r_2<r_{\rm out}$ such that~$\psi^\prime(r_2)=0$ and~$\psi^\prime(r)>0$ for~$r_2<r<r_{\rm out}$.  Therefore, by Lemma~\ref{lem:psi_re}, 
\begin{equation}\label{eq:psi_r2}
\psi(r_2)\ge f(r_2).
\end{equation}
If~$r_2=r_{\rm in}$, the solution satisfies~\eqref{eq:lem4res}, see also Figure~\ref{fig:Lemma4Illustration}.  Otherwise,~$r_2>r_{\rm in}$.  If~$r_2$ is a local minimum point of~$\psi$, then there exists a maximum point of~$\psi$ from the right at a point~$r_{\rm in}\le r_M<r_2$ that satisfies~$\psi(r_M)>f(r_M)$.  In contradiction with Lemma~\ref{lem:psi_re}.  Therefore,~$r_2$ is an inflection point and by Lemma~\ref{lem:psi_re},~$\psi(r_2)=f(r_2)$ .  In this case, repeating the argument, the boundary condition~$\psi^\prime(r_{\rm in})=0$ ensures that there exists a point~$r_{\rm in}\le r_3<r_2$ such that~$\psi^\prime(r_3)=0$ and~$\psi^\prime(r)<0$ for~$r_3<r<r_2$.  If~$r_3=r_{\rm in}$, the solution satisfies~\eqref{eq:lem4res}.  Otherwise, the argument is repeated to yields a sequence~$r_2>r_3>\cdots>r_k>r_{\rm in}$ of such critical points.  Differentiability of~$\psi$ ensures that the sequence is finite.  The solution satisfies~$\psi^\prime>0$ for~$r_{\rm in}<r<r_{\rm out}$, except at the critical points where~$\psi^\prime(r_j)=0$.  Further, at~$r=r_{\rm in}$ the solution satisfies~\eqref{eq:psi_r2}.  Therefore, conditions~\eqref{eq:lem4res} are satisfied.
\end{proof}
\section{The fully superconductive case}\label{sec:superconductive}
The works~\cite{kapon2018nature,kapon2019phase} consider the case~$\psi\equiv1$, namely the case for which the whole cylinder region is in a superconductive state.   However,~$\psi\equiv1$ is not a solution of~\eqref{eq:system_ring_region}.  Indeed, by~\eqref{eq:psi}, 
$\psi\equiv1$ implies~$A_J\equiv0$.  But~$A_J\equiv0$ does not satisfy the boundary condition at~$r=r_{\rm in}$ for~$J\ne0$, see~\eqref{eq:A}.  

Here, we focus on the case~$\psi\approx1$ such that to leading order Equation~\eqref{eq:A} for~$A_J$ reduces to
\begin{equation}\label{eq:A_linear}
A_J^{\prime\prime}(r)+\frac{A_J^\prime}r-\frac{A_J}{r^2}=\frac{1}{\lambda^2}A_J(r),\qquad A_J^\prime(r_{\rm in})-\frac{A_J}{r_{\rm in}}=-\frac{2}{r_{\rm in}^2},\quad A^\prime_J(r_{\rm out})+\frac{A_J}{r_{\rm out}}=0.
\end{equation}
The exact solution of~\eqref{eq:A_linear} is given by the profile~$A_J=B(r)$ defined in terms of Bessel functions
\begin{subequations}\label{eq:A_leadingOrder_cxplicit}
\begin{equation}
B(r;\lambda,r_{\rm in},r_{\rm out})=c_1\left(\frac{r_{\rm in}}{\lambda},\frac{r_{\rm out}}{\lambda}\right) I_1\left(\frac{r}{\lambda}\right)+c_2\left(\frac{r_{\rm in}}{\lambda},\frac{r_{\rm out}}{\lambda}\right) K_1\left(\frac{r}{\lambda}\right),
\end{equation}
where~$c_1,c_2$ are determined by the boundary conditions
\begin{equation}
c_1=-\frac1{r_{\rm in}^2}\frac{2\lambda K_0\left(\frac{r_{\rm out}}{\lambda}\right)}{ I_2\left(\frac{r_{\rm rin}}{\lambda}\right)K_0\left(\frac{r_{\rm out}}{\lambda}\right) -K_2\left(\frac{r_{\rm in}}{\lambda}\right)I_0\left(\frac{r_{\rm out}}{\lambda}\right)},
\end{equation}
and
\begin{equation}
c_2=\frac{I_0\left(\frac{r_{\rm out}}{\lambda}\right)}{K_0\left(\frac{r_{\rm out}}{\lambda}\right)}c_1.
\end{equation}
\end{subequations}
Note that the profile~$B$ does not depend on~$J$.  Therefore, in the low flux regime, the vector potential~$A_{\rm sc}$ due to the superconducting ring takes the form
\begin{equation}\label{eq:Asc_J_dependence}
\frac{A_{\rm sc}}{J}=B(r)-\frac1r
\end{equation}
and particularly scales linearly with~$J$.

We have seen that~$\psi\equiv1$ and~$A_J\equiv B$ satisfies equation~\eqref{eq:A}.  Substituting the approximation~$A_J\approx B$ into~$\eqref{eq:psi}$, implies that at~$r=r_{\rm in}$,
\[
\psi^{\prime\prime}(r_{\rm in})\approx J^2B^2(r_{\rm in}).
\]
Therefore,~$\psi\equiv1$ does not satisfy~\eqref{eq:psi} near~$r=r_{\rm in}$. This suggests a boundary layer at~$r=r_{\rm in}$.  
The equation for~$\psi$ \eqref{eq:psi} suggests that the width of the boundary layer is~$O(\varepsilon)$.
We distinguish between the solution~$\psi_{\rm outer}$ outside the boundary layer, aka the outer solution, and the solution in the boundary layer region, aka the inner solution,~$\psi_{\rm in}$.
Substituting the approximation~$A_J\approx B$ into~$\eqref{eq:psi}$, the outer solution satisfies
\[
\psi_{\rm outer}^3(r)-(1-\varepsilon^2J^2B^2(r))\psi_{\rm outer}=0,
\]
or
\[
\psi_{\rm outer}=\sqrt{1-\varepsilon^2J^2B^2(r)}.
\]
Let us seek for an inner solution of the form
\[
\psi_{\rm in}(\rho_\varepsilon)=\sqrt{1-\varepsilon^2J^2B^2(r)}+f(\rho_\varepsilon),\quad \rho_\varepsilon=\frac{r-r_{\rm in}}{\varepsilon},\quad |f|\ll 1.
\]
Substituting~$\psi_{\rm in}(\rho_\varepsilon)$ into~\eqref{eq:psi} yields
\begin{equation}\label{eq:psi_inner_ansatz_f}
f^{\prime\prime}(\rho_\varepsilon)-2f(\rho_\varepsilon)=O(f^2,\varepsilon f,\varepsilon^4),\quad f^\prime(0)=-\psi_{\rm outer}^\prime(0),\quad f(\infty)=0
\end{equation}
where the former boundary conditions assures ~$\psi_{\rm in}^\prime(0)=0$ and the latter condition is the Prandtl matching condition.  Thus, 
\begin{equation}\label{eq:psi_inner}
\psi_{\rm in}=\sqrt{1-\varepsilon^2J^2B^2}+c\,\varepsilon^3  e^{-\sqrt2\rho_\varepsilon},\quad c=-J^2\frac{\sqrt{2}}{2}\frac{B(r_{\rm in})B^\prime(r_{\rm in})}{\sqrt{1 - \varepsilon^2J^2B(r_{\rm in})^2}}=-\frac{\sqrt{2}}{2}\frac{J^2B(r_{\rm in})}{\sqrt{1 - \varepsilon^2J^2B^2(r_{\rm in})}}\left[\frac{B(r_{\rm in})}{r_{\rm in}}-\frac{2}{r_{\rm in}^2}\right].
\end{equation}

The approximation error in~$\psi_{\rm in}$ is~$O(f^2,\varepsilon f)$, see~\eqref{eq:psi_inner_ansatz_f}.  By~\eqref{eq:psi_inner},~$f=O(\varepsilon^3)$.  Therefore,
\begin{subequations}\label{eq:profilesApprox}
\begin{equation}\label{eq:psi_approx}
\psi^{\rm approx}=\sqrt{1-\varepsilon^2J^2B^2(r)}-\frac{\sqrt{2}\,\varepsilon^3 J^2B(r_{\rm in})}{2\sqrt{1 - \varepsilon^2J^2B^2(r_{\rm in})}}\left[\frac{B(r_{\rm in})}{r_{\rm in}}-\frac{2}{r_{\rm in}^2}\right]\,e^{-\sqrt2\rho_\varepsilon}+O(\varepsilon^4).
\end{equation}
The reduction from~\eqref{eq:A} for~$A_J$ to~\eqref{eq:A_linear} for~$B$ relied on the approximation~$\psi^2\approx 1$.  Approximation~\eqref{eq:psi_approx} for~$\psi$ implies that 
\[
\psi=\sqrt{1-\varepsilon^2J^2B^2(r)}+O(\varepsilon^3)=1-\frac12\varepsilon^2J^2B^2(r)+O(\varepsilon^3).
\]
Namely, that the reduction from~\eqref{eq:A} to~\eqref{eq:A_linear} introduced an~$O(\varepsilon^2)$ error.  Hence,
\begin{equation}\label{eq:AJ_approx}
A_J^{\rm approx}(r)=B(r;\lambda,r_{\rm in},r_{\rm out})+O(\varepsilon^2),
\end{equation}
where~$B$ is given by~\eqref{eq:A_leadingOrder_cxplicit}.
\end{subequations}
Note that the approximation~\eqref{eq:AJ_approx} was not derived in the~$\lambda\ll1$, and is valid to larger~$\lambda$.  This large region of validity  is attained since~\eqref{eq:A_linear} has an explicit solution.  In the subsequent sections, we will present cases in which an explicit solution to the leading order equation is not available, and further approximation utilizing~$\lambda\ll1$ is required.   To address these cases, it is helpful to also consider an approximation of~$B(r)$ and~$\psi^{\rm approx}$ for~$\lambda\ll1$.
An asymptotic expansion of~$B(r)$~\eqref{eq:A_leadingOrder_cxplicit} for~$\lambda\ll1$ yields
\begin{equation}\label{eq:A_asympt}
B^{\rm asympt}=\frac{2\lambda}{\sqrt{r_{\rm in}^{3}r}}e^{-\rho}\left[1-\frac38\frac{5 r - r_{\rm in}}{r_{\rm in}r}\lambda+\frac{15}{128}\frac{23r^2-6rr_{\rm in}-r_{\rm in}^2}{r_{\rm in}^2r^2}\lambda^2+O(\lambda^3)\right].
\end{equation}
Substituting~\eqref{eq:A_asympt} in~\eqref{eq:profilesApprox} gives rise to the leading order approximation
\begin{equation}\label{eq:leadingOrderSmallJ}
A_J=\frac{2\lambda}{r_{\rm in}^2}e^{-\rho}+O(\lambda^2),\qquad \psi=1-\frac{2\lambda^2\varepsilon^2 J^2}{r_{\rm in}^4}e^{-2\rho}+\frac{2\sqrt{2}\lambda \varepsilon^3J^2}{r_{\rm in}^4}e^{-\sqrt2\rho_\varepsilon}+O(\varepsilon^2\lambda^3,\varepsilon^3\lambda^2),
\end{equation}
where
\[\rho=\frac{r-r_{\rm in}}{\lambda},\quad \rho_\varepsilon=\frac{r-r_{\rm in}}{\varepsilon}.\] 
Particularly,~$\psi$ has a double boundary layer at~$r=r_{\rm in}$: A boundary layer of width~$\lambda/2$, and an internal  layer of width~$\varepsilon/\sqrt{2}$.  
\subsection{Numerical verification}
We now present a numerical verification of the asymptotic results presented in this section.
In Figure~\ref{fig:lowCurrentRegime_profiles} we compare between a numerical solution of~\eqref{eq:system_ring_region} and its corresponding approximation~\eqref{eq:profilesApprox} for~$\varepsilon=5\cdot 10^{-3}$.  The approximation error for~$\psi$ is~$O(\varepsilon^4)$ where in this case~$\varepsilon^4\approx 6\cdot10^{-10}$.  Therefore, as expected, the curves~$\psi$ and~$\psi^{\rm approx}$ are indistinguishable, see Figure~\ref{fig:lowCurrentRegime_profiles}A.  Similarly, the curves~$A_J$ and~$A_J^{\rm approx}$ are indistinguishable, see Figure~\ref{fig:lowCurrentRegime_profiles}B.  We further focus on the internal boundary layer region~$r-r_{\rm in}=O(\varepsilon)$, see inset graphs in Figure~\ref{fig:lowCurrentRegime_profiles}, and observe that as expected~$\psi$ deviates from the outer solution~$\psi_{\rm outer}$ in the boundary layer, but~$A_J$ does not deviate from~$B$ in this region. 
\begin{figure}[ht!]
\begin{center}
\scalebox{0.75}{\includegraphics{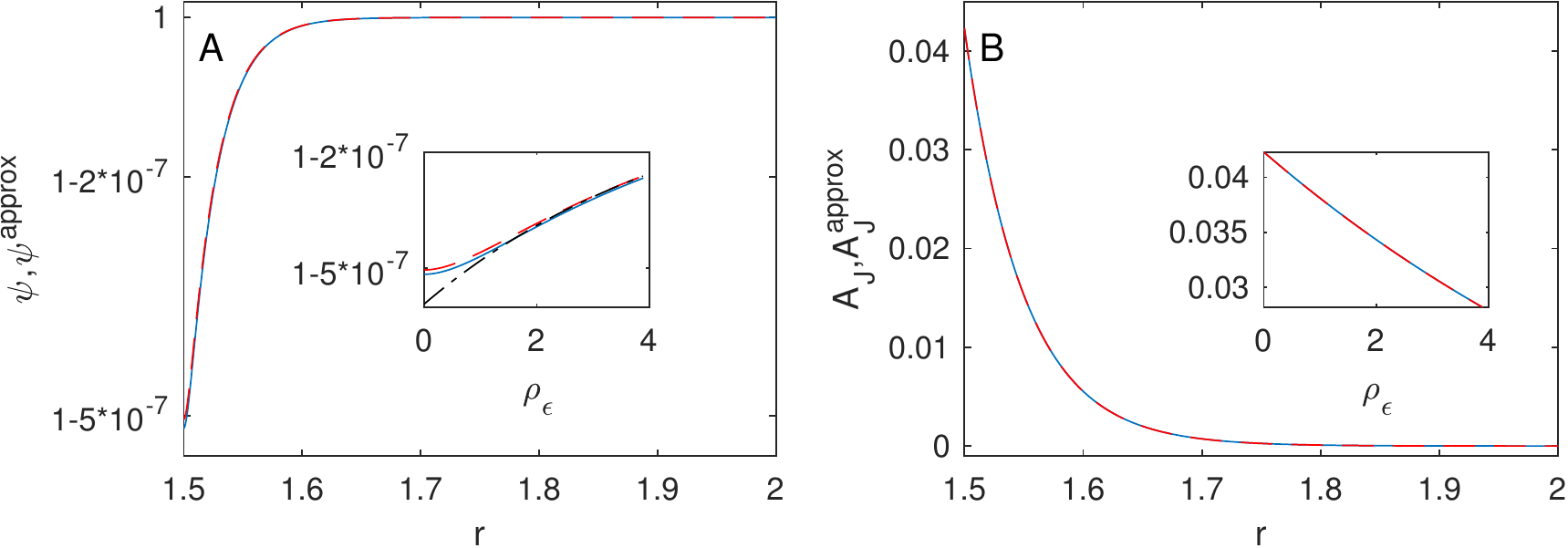}}
\caption{Numerical solution of~\eqref{eq:system_ring_region} for~$r_{\rm in}=1.5$,~$r_{\rm out}=2$,~$\varepsilon=0.005$, and~$\lambda=0.05$ ({\color{blue} solid}), as well as the corresponding approximation~\eqref{eq:profilesApprox}. ({\color{red} dashes}).  Curves are indistinguishable.  Insets present the same data, in the boundary layer region~$r_{\rm in}<r<r_{\rm in}+\varepsilon \rho_\varepsilon$.  Super-imposed in the inset of A is the outer solution approximation~$\psi_{\rm outer}$ (dash-dots). Graph A:~$\psi$.  Graph B:~$A_J$.  }
\label{fig:lowCurrentRegime_profiles}
\end{center}
\end{figure}

Next we consider the approximation errors~$E=\|\psi-\psi^{\rm approx}\|_{\infty}$ and~$\|A_J-A_J^{\rm approx}\|_{\infty}$, where~$\psi,\,A_J$ are numerical solutions of~\eqref{eq:system_ring_region}, and~$\psi^{\rm approx},\,A_J^{\rm approx}$ are the corresponding approximations~\eqref{eq:profilesApprox}, respectively.  We observe that, as expected by~\eqref{eq:profilesApprox},~$\|\psi-\psi^{\rm approx}\|_\infty=O(\varepsilon^4)$ and~$\|A_J-A_J^{\rm approx}\|_\infty=O(\varepsilon^2)$, see Figures~\ref{fig:lowCurrentRegime_E_vs_epsilon}A and ~\ref{fig:lowCurrentRegime_E_vs_epsilon}B, respectively. 
\begin{figure}[ht!]
\begin{center}
\scalebox{0.75}{\includegraphics{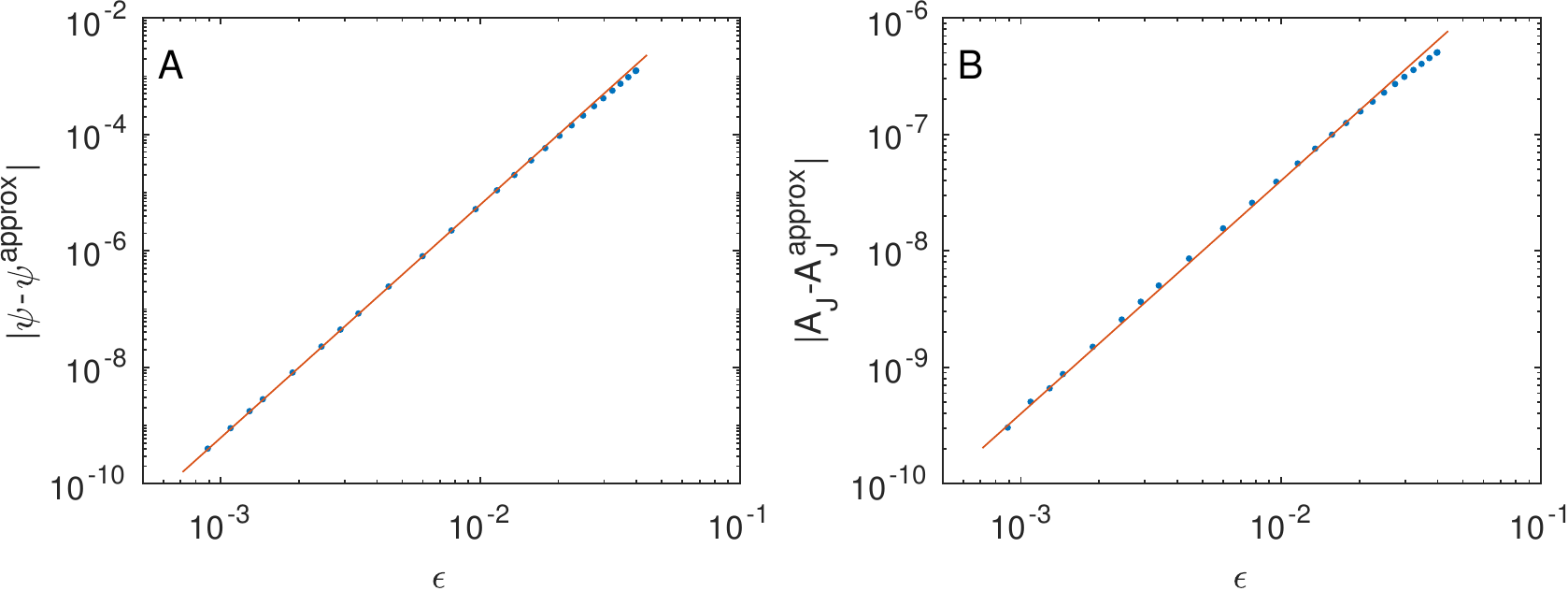}}
\caption{Approximation error~$\|\psi-\psi^{\rm approx}\|_\infty$ (graph A) and~$\|A_J-A_J^{\rm approx}\|_\infty$ (graph B) as a function of~$\varepsilon$ where~$\psi$ and~$A_J$ are numerical solutions of~\eqref{eq:system_ring_region} for~$r_{\rm in}=1.5$,~$r_{\rm out}=2$ and~$\lambda=0.05$ ({\color{blue} dots}), and~$\psi^{\rm approx}$ and~$A_J^{\rm approx}$ are the corresponding approximations~\eqref{eq:profilesApprox}. Super-imposed in graph A is the curve~$620\,\varepsilon^4$ and in graph B the curve~$0.0004\,\varepsilon^2$ ({\color{red} solid}).}
\label{fig:lowCurrentRegime_E_vs_epsilon}
\end{center}
\end{figure}
\subsection{Emerging picture - low flux regime}
\begin{figure}[ht!]
\begin{center}
\scalebox{0.5}{\includegraphics{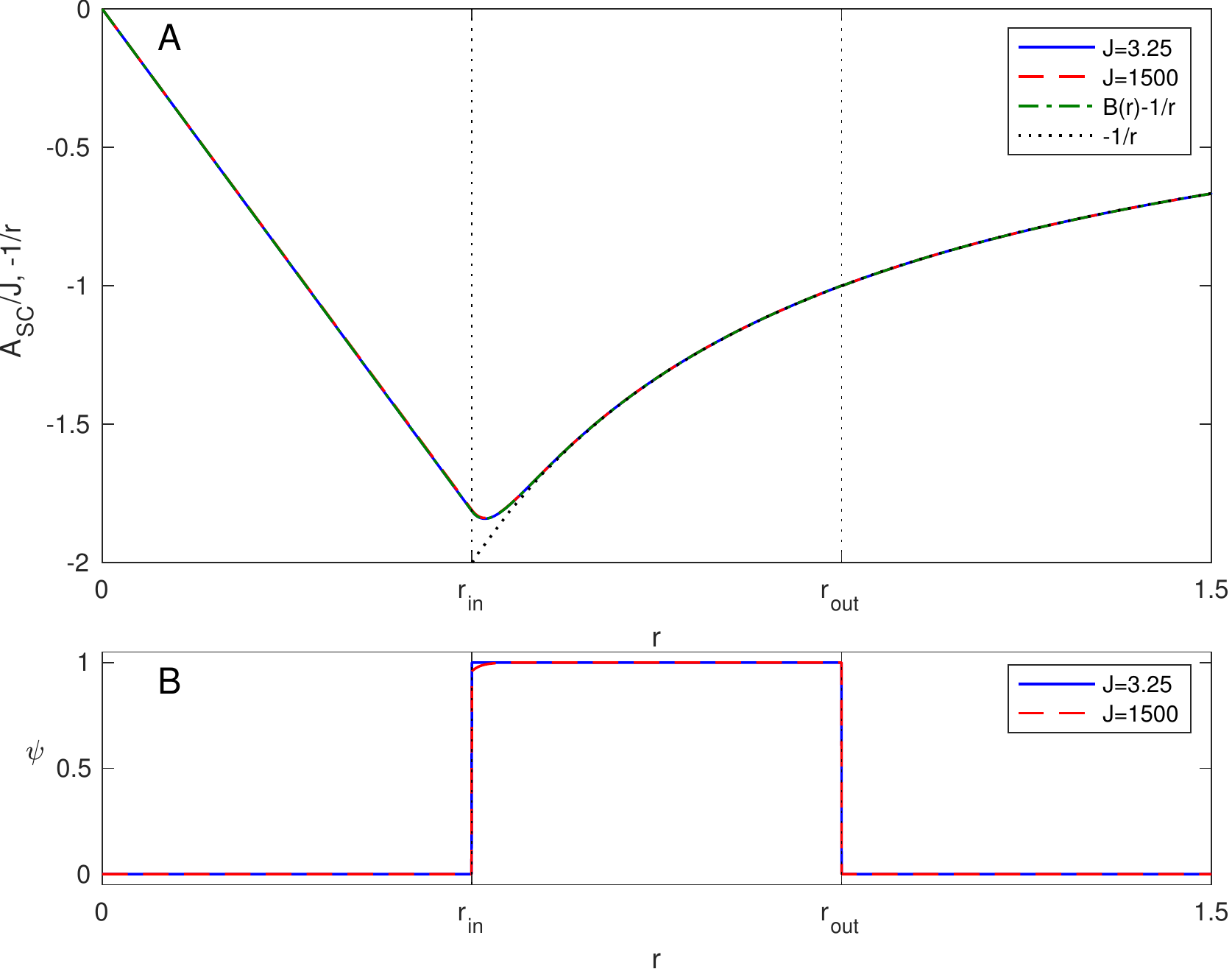}}
\caption{A: Graph of scaled vector potential~$A_{\rm sc}/J=A_J-1/r$, where~$A_J$ is the numerical solution of system~\eqref{eq:system_ring_region} for~$r_{\rm in}=0.5$,~$r_{\rm out}=1$,~$\varepsilon=0.001$, and~$\lambda=0.025$ for~$J=3.25$ ({\color{blue} solid}), ~$J=1500$ ({\color{red} dashes}).  Super-imposed is the profile~$B(r;r_{\rm in}=0.5,r_{\rm out}=1,\lambda=0.025)-1/r$ given by \eqref{eq:A_leadingOrder_cxplicit} ({\color{greenCurve} dash-dots}) and the function~$-1/r$ (dots).  The first three curves are indistinguishable.  B: Graph of order parameter~$\psi$ corresponding to the two numerical solutions of~\eqref{eq:system_ring_region} presented in A.  The two curves are distinguishable only near~$r=r_{\rm in}$.}
\label{fig:generalPictureAJ_lowCurrent}
\end{center}
\end{figure}It is instructive to consider the results of the above analysis in terms of the vector potential~$A_{\rm sc}$ due to the superconductor.  At low flux regimes, the scaled profile~$A_{\rm sc}/J$ is shown to be independent of~$J$ and well approximated by the profile $B$ given by \eqref{eq:A_leadingOrder_cxplicit} via the relation~\eqref{eq:Asc_J_dependence}.  To demonstrate this, in Figure~\ref{fig:generalPictureAJ_lowCurrent}A we plot the scaled profiles~$A_{\rm sc}/J=A_J-1/r$ where~$A_J$ are the numerical solutions of system~\eqref{eq:system_ring_region} for~$J=3.25$ and~$J=1500$, as well as the profile~$B(r;r_{\rm in},r_{\rm out},\lambda)-1/r$ with appropriate (non-fitted) parameters.  In Figure~\ref{fig:generalPictureAJ_lowCurrent}B we plot the profiles~$\psi$ for the two case, and observed that both solutions correspond to fully super-conductive cases, namely the current~$J=1500$ is within the low flux regime.  The three curves in Figure~\ref{fig:generalPictureAJ_lowCurrent}A are indistinguishable showing that the dependence of scaled vector potential~$A_{\rm sc}/J$ on the current~$J$ is negligible at low flux regimes.  Note that the term `low flux' will be quantified in the subsequent section.  

It is instructive to consider the approximation~\eqref{eq:A_leadingOrder_cxplicit} in a high flux regime which is beyond the expected region of validity of the approximation, as will be demonstrated subsequently.  In Figure~\ref{fig:generalPictureAJ_partialSC}A, we plot the scaled profiles~$A_{\rm sc}/J=A_J-1/r$ where~$A_J$ are the numerical solutions of system~\eqref{eq:system_ring_region} for~$J=3500$, as well as the profile~$B(r;r_{\rm in},r_{\rm out},\lambda)-1/r$ with appropriate (non-fitted) parameters.  In Figure~\ref{fig:generalPictureAJ_partialSC}B we plot the profile~$\psi$, and observe only partial super-conductivity near the inner rim of the cylinder.  This implies that the considered case is beyond the validity region of approximation~\eqref{eq:A_leadingOrder_cxplicit}.   Nevertheless, in Figure~\ref{fig:generalPictureAJ_partialSC}A we observe a fair agreement between the solution profile and the corresponding approximation~\eqref{eq:A_leadingOrder_cxplicit}.   
\begin{figure}[ht!]
\begin{center}
\scalebox{0.5}{\includegraphics{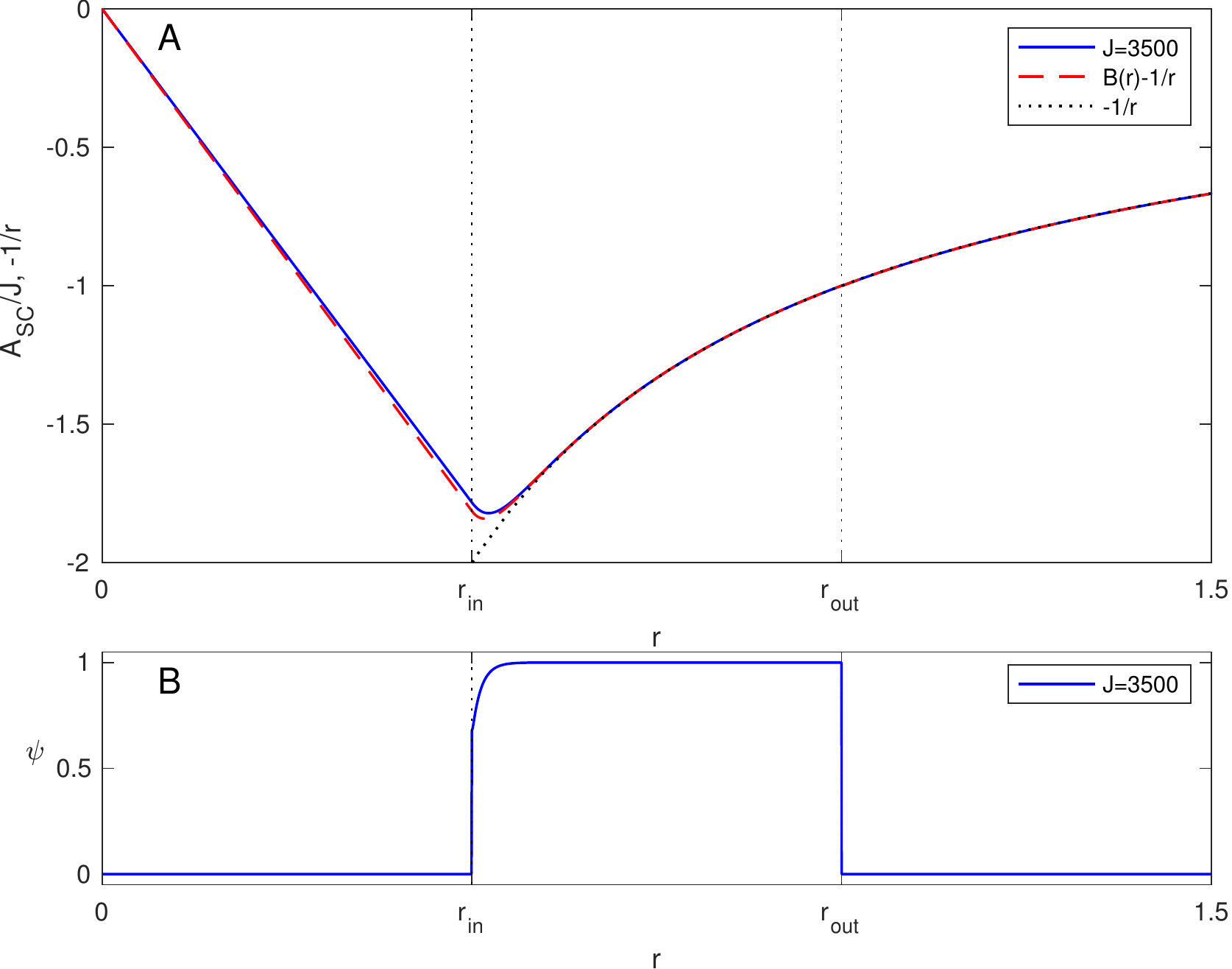}}
\caption{A: Graph of scaled vector potential~$A_{\rm sc}/J=A_J-1/r$, where~$A_J$ is the numerical solution of system~\eqref{eq:system_ring_region} for~$r_{\rm in}=0.5$,~$r_{\rm out}=1$,~$\varepsilon=0.001$, and~$\lambda=0.025$ for~$J=3500$ ({\color{blue} solid}).  Super-imposed is the profile~$B(r;r_{\rm in}=0.5,r_{\rm out}=1,\lambda=0.025)-1/r$ given by \eqref{eq:A_leadingOrder_cxplicit} ({\color{red} dashes}) and the function~$-1/r$ (dots).  B: Graph of order parameter~$\psi$ corresponding to the numerical solution of~\eqref{eq:system_ring_region} presented in A.}
\label{fig:generalPictureAJ_partialSC}
\end{center}
\end{figure}

As discussed in the presentation of the reduction from system~\eqref{eq:system} to system~\eqref{eq:system_ring_region}, the solution in the hollow cylinder region~$[r_{\rm in},r_{\rm out}]$ undergoes a transition from a linear graph~$c_1r$ to~$c_2/r$.  The analysis in Section~\ref{sec:superconductive} better characterizes the nature of this transition: In terms of the scaled vector potential~$A_{\rm sc}/J$, the solution undergoes a transition from~$c_1r$ to~$-1/r$, and this transition occurs in the narrow transition layer of~$O(\lambda)$ width at~$r=r_{\rm in}$.  The curve~$-1/r$ is super-imposed in Figure~\ref{fig:generalPictureAJ_lowCurrent}A.  As expected, one can observe that after a small transition layer the scaled  profiles well agree with the curve~$-1/r$.  From a physical point of view, this means that the superconductor generates currents in a layer of width~$\lambda$ from its inner rim. These currents produce a flux in the hole that exactly cancels the applied flux. Consequently, deep in the superconductor, the total vector potential~$A_J$ is zero, and there is no current or magnetic field.

We note that the original system~\eqref{eq:system} can be expressed in terms of numerous relevant physical quantities, e.g.,~$A_J$ with a choice of scaling or~$A_{\rm sc}$.  In this work, we have introduced the choice~\eqref{eq:AJ}.  It is instructive to refer to Figure~\ref{fig:generalPictureAJ_lowCurrent} for visual motivation and support of this choice.  Indeed, Figure~\ref{fig:generalPictureAJ_lowCurrent} strongly suggests that one should consider the scaled vector potential~$A_{\rm sc}/J$ since it is independent of~$J$ in low flux regimes.  Furthermore, since~$A_{\rm sc}/J\approx -1/r$ in most of the domain, it is preferable to study the quantity~$A_{\rm sc}/J+1/r$.  This is exactly the quantity~$A_J$, see~\eqref{eq:AJ}.

\section{High flux regime}\label{sec:highCurrentRegime}
Section~\ref{sec:superconductive} focused on the fully superconductive case for which~$\psi\approx1$ in the cylinder region.
Approximation~\eqref{eq:leadingOrderSmallJ} of~$\psi$ implies that this analysis is valid in the parameter regime, aka, the low flux regime, 
\begin{equation}\label{eq:low_current_regime}
\frac{\lambda^2\varepsilon^2J^2}{r_{\rm in}^4}\ll1\quad \mbox{or}\quad J\ll\frac{r_{\rm in}^2}{\lambda\varepsilon}.
\end{equation}
In this section, we consider the high flux regime
\begin{equation}\label{eq:J_highcurrent_regime}
J=O\left(\frac{1}{\lambda\varepsilon}\right).
\end{equation}
Similar to the analysis presented in Section~\ref{sec:superconductive}, the equation for~$\psi$, see~\eqref{eq:psi}, suggests a boundary layer of width~$\varepsilon$ at~$r=r_{\rm in}$.  The outer solution satisfies, to leading order,
\begin{equation}\label{eq:psi_outer_cq}
\psi_{\rm outer}^3(r)-(1-\varepsilon^2J^2A_J^2(r))\psi_{\rm outer}=0.
\end{equation}
This equation has the solution~$\psi_{\rm outer}=0$ and in the case~$1-\varepsilon^2J^2A_J^2(r)>0$ for~$r_{\rm in}<r<r_{\rm out}$ a second solution
\begin{equation}\label{eq:psi_outer}
\psi_{\rm outer}=\sqrt{1-\varepsilon^2J^2A_J^2(r)}.
\end{equation}
Since we expect a continuous change in the behavior as the current~$J$ is increased, and since~$\psi\approx1$ in the low current region, see Section~\ref{sec:superconductive}, we will now study the case of partial superconductivity where~$1-\varepsilon^2J^2A_J^2(r)>0$ and~$\psi_{\rm outer}$ is given by~\eqref{eq:psi_outer}.  Additional cases will be considered in subsequent sections.

\subsection{Partial superconductivity}\label{sec:partialSC}
Let us consider the case 
\begin{equation}\label{eq:inequality_AJ}
1-\varepsilon^2J^2A_J^2(r)>0
\end{equation}
 for~$r_{\rm in}<r<r_{\rm out}$, for which~$\psi_{\rm outer}$ is given, to leading order, by~\eqref{eq:psi_outer}.
Substituting~$\psi_{\rm outer}$~\eqref{eq:psi_outer} in~\eqref{eq:A} yields
\begin{equation}\label{eq:A_nonlinear}
A_J^{\prime\prime}(r)+\frac{A_J^\prime}r-\frac{A_J}{r^2}=\frac{1}{\lambda^2}\left[1-\varepsilon^2J^2 A^2_J(r)\right]A_J(r),\qquad A_J^\prime(r_{\rm in})-\frac{A_J}{r_{\rm in}}=-\frac{2}{r_{\rm in}^2},\quad A^\prime_J(r_{\rm out})+\frac{A_J}{r_{\rm out}}=0.
\end{equation}
The function~$A_J=0$ satisfies the above equation and the boundary condition at~$r=r_{\rm out}$, but does not satisfy the boundary condition at~$r=r_{\rm in}$.  This suggests a boundary layer in~$r=r_{\rm in}$ with an outer solution~$A_J=0$.
To study the vector potential~$A_J$ in the boundary layer region, let us consider the scaled current
\begin{equation}\label{eq:Js}
J_s=\varepsilon \lambda J,
\end{equation}
and the scaled quantities in the boundary layer regime of~$A_J$ for the inner solution
\begin{equation}\label{eq:As}
A_{\rm in}=\frac{A_J}{\lambda},\quad \rho=\frac{r-r_{\rm in}}{\lambda},
\end{equation}
where the scaling of~$A_J$ is since~$1-\varepsilon^2J^2A_J^2(r)>0$ and~\eqref{eq:J_highcurrent_regime} imply that~$A_J=O(\lambda)$.

Substituting~\eqref{eq:Js} and~\eqref{eq:As} in~\eqref{eq:A_nonlinear}, and using the Prandtl matching condition for the outer solution~$A_{\rm in}^{\rm outer}=0$, gives rise to the equation for the scaled vector potential~$A_{\rm in}$
\begin{equation}\label{eq:As_full}
A_{\rm in}^{\prime\prime}(\rho)+\lambda\frac{A_{\rm in}^\prime}{r_{\rm in}+\lambda\rho}-\lambda^2\frac{A_{\rm in}}{(r_{\rm in}+\lambda\rho)^2}=A_{\rm in}-J_s^2 A_{\rm in}^3,\qquad A_{\rm in}^\prime(0)=\lambda\frac{A_{\rm in}(0)}{r_{\rm in}}-\frac{2}{r_{\rm in}^2},\quad A_{\rm in}(\infty)=0.
\end{equation}
Equation~\eqref{eq:As_full} is a nonlinear equation and, to the best of our knowledge, does not have an explicit analytic solution.  This is in contrast to the low current case~\eqref{eq:low_current_regime} for which the corresponding equation~\eqref{eq:A_linear} is linear and can be solved explicitly.  Let us seek for an solution of~\eqref{eq:As_full} for~$\lambda\ll1$ in the form
\begin{equation}\label{eq:As_expansion}
A_{\rm in}(\rho)=A_0(\rho)+\lambda A_1(\rho)+O(\lambda^2),
\end{equation}
Substituting~\eqref{eq:As_expansion} in~\eqref{eq:As_full} and equating the~$O(1)$ and~$O(\lambda)$ terms yields 
\begin{equation}\label{eq:As_leadingOrder}
A_0^{\prime\prime}(\rho)-A_0=-J_s^2A_0^3,\qquad A_0^\prime(0)=-\frac{2}{r_{\rm in}^2},\quad A_0(\infty)=0,
\end{equation}
and
\begin{equation}\label{eq:As_firstOrder}
A_1^{\prime\prime}(\rho)-A_1=-\frac{A_0^\prime}{r_{\rm in}}-3J_s^2A_0^2A_1,\qquad A_1^\prime(0)=\frac1{r_{\rm in}}A_0(0),\quad A_1(\infty)=0,
\end{equation}
respectively.

We now solve~\eqref{eq:As_leadingOrder} for the leading order solution~$A_0$.  Multiplying both hands of~\eqref{eq:As_leadingOrder} by~$A_0^\prime$ and integrating in~$\rho$ while using~$A_0(\infty)=0$ yields
\begin{equation}\label{eq:nonlinear_A}
[A_0^\prime(\rho)]^2=A_0^2-\frac{J_s^2}{2}A_0^4.
\end{equation}
Lemma~\ref{lem:AJ} implies that the solution is positive and monotonically decreasing.  Therefore, we consider the  branch
\[
A_0^\prime(\rho)=-\sqrt{A_0^2-\frac{J_s^2}{2}A_0^4}.
\]
The inverse function~$\rho(A_0)$ satisfies
\[
\rho^\prime(A_0)=-\frac{1}{\sqrt{A_0^2-\frac{J_s^2}{2}A_0^4}}.
\]
Integration of~$\rho^\prime(A_0)$ while using~$\rho(A_0(0))=0$ yields
\begin{equation}\label{eq:rho_A0}
\rho(A_0)=\frac{1}2\ln\left[\alpha\frac{1+\sqrt{1-\frac{J_s^2}{2}A_0^2}}{1-\sqrt{1-\frac{J_s^2}{2}A_0^2}}\right],\quad \alpha=\frac{1-\sqrt{1-\frac{J_s^2}{2}A_0^2(0)}}{1+\sqrt{1-\frac{J_s^2}{2}A_0^2(0)}}.
\end{equation}
The constant~$\alpha$ depends on the unknown value~$A_0(0)$.  To resolve~$\alpha$, we substitute the boundary condition for~$A_0$ at~$\rho=0$ in~$\eqref{eq:nonlinear_A}$ which implies that~$A_0^2(0)$ equals one of two values~$c_\pm$
\begin{equation}\label{eq:A0_pm}
A_0^2(0)=c_\pm=\frac{1 \pm \sqrt{1 -8\frac{J_s^2}{r_{\rm in}^4}}}{J_s^2}.
\end{equation}

However, according to \eqref{eq:Js} and~\eqref{eq:A0_pm},~$1-J_s^2 c_+<0$ which contrasts with~\eqref{eq:inequality_AJ}.  Hence,
\begin{equation}\label{eq:A20}
A_0^2(0)=\frac{1 - \sqrt{1 -8\frac{J_s^2}{r_{\rm in}^4}}}{J_s^2}.
\end{equation}
Isolating~$A_0$ in~\eqref{eq:rho_A0} and substituting~\eqref{eq:A20} yields
\begin{subequations}\label{eq:approxAs}
\begin{equation}\label{eq:A0_result}
A_0(\rho)=\frac{\sqrt{8\alpha}}{J_s}\frac{e^{-\rho}}{1+\alpha\,e^{-2\rho}},\quad \alpha=\frac{1-\sqrt{\frac12\left[1+\sqrt{1 -8\frac{J_s^2}{r_{\rm in}^4}}\right]}}{1+\sqrt{\frac12\left[1+\sqrt{1 -8\frac{J_s^2}{r_{\rm in}^4}}\right]}}.
\end{equation}
Substituting~$A_0(\rho)$ in~\eqref{eq:As_firstOrder} for~$A_1(\rho)$, and solving it yields
\begin{equation}\label{eq:A1_result}
A_1=\frac{\sqrt{2\alpha}}{3\,J_s\,r_{\rm in}\, (1+\alpha\,e^{-2\rho})^2 }\left[9 e^{-\rho}+\frac {\left( {\alpha}^{3}+11\,{\alpha}^{2}-69\,\alpha-15
 \right)(e^{-\rho}-\alpha e^{-3\rho}) }{\left( {\alpha}^{2}-6\,\alpha+1
 \right) }-6(e^{-\rho}-\alpha e^{-3\rho})(2+\rho) -\alpha^2 e^{-5\rho}
\right].
\end{equation}
\end{subequations}
We now consider the approximation error.  The reduction from~\eqref{eq:A} to~\eqref{eq:As_full} relied on the approximation of~$\psi$ by its outer solution~$\psi_{\rm outer}$.  Namely, neglecting the possible contribution to the vector potential~$A_J$ due to a boundary layer of~$\psi$ at~$r=r_{\rm in}$.  Then, the outer solution~$\psi_{\rm outer}$ is approximated to leading order, see~\eqref{eq:psi_outer}.
To quantify the errors involved in the above approximation, let us compute~$\psi$ more accurately.
As in the analysis of the low current case, see Section~\ref{sec:superconductive}, we consider the ansatz
\[
\psi(\rho,\rho_\varepsilon)=\sqrt{1-J_s^2A_{\rm in}^2(\rho)}+f(\rho_\varepsilon),\quad \rho_\varepsilon=\frac{r-r_{\rm in}}{\varepsilon},\quad |f|\ll 1.
\]
Substituting~$\psi(\rho,\rho_\varepsilon)$ into~\eqref{eq:psi} and solving the equation for~$f$ yields, see details in Section~\ref{sec:superconductive} and particularly before equation~\eqref{eq:psi_inner_ansatz_f}:
\begin{equation}\label{eq:psi_including_inner}
\psi=\sqrt{1-J_s^2A_{\rm in}^2(\rho)}+\frac{\varepsilon}{\lambda}\beta \exp\left[-\sqrt2\sqrt{1-J_s^2A_{\rm in}^2(0)}\rho_\varepsilon\right]+O\left(\frac{\varepsilon^2}{\lambda^2}\right), \quad \beta=-\frac{\sqrt2J_s^2}{2}\frac{A_{\rm in}^\prime(0)A_{\rm in}(0)}{1-J_s^2A_{\rm in}^2(0)},
\end{equation}
where~$A_{\rm in}(\rho)$ is the solution of equation~\eqref{eq:As_full}.   Finally, we note that one can readily substitute approximation~\eqref{eq:approxAs} of~$A_{\rm in}$ in~\eqref{eq:psi_including_inner} to obtain an approximation of~$\psi$ that depends only on the problem parameters.  

\begin{figure}[ht!]
\begin{center}
\scalebox{0.5}{\includegraphics{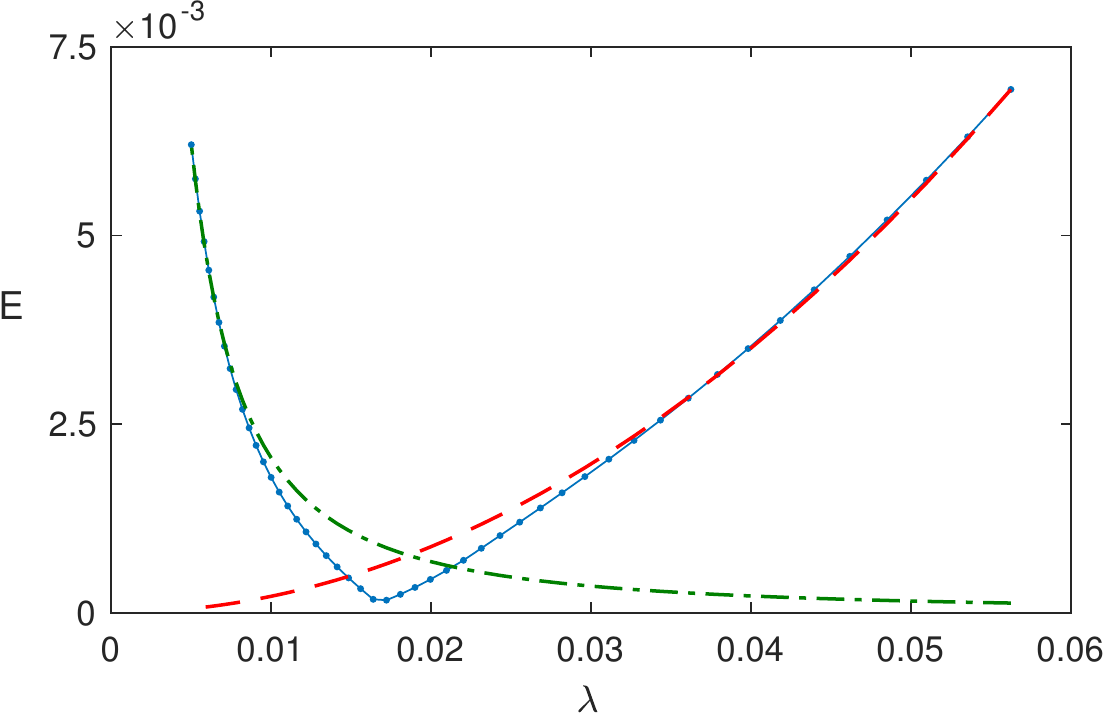}}
\caption{Error~$E=\max|A_0(\rho)+\lambda A_1(\rho)-A_{\rm in}^{\rm numerical}(\rho;\lambda)|$ where~$A_{\rm in}^{\rm numerical}(\lambda)=A_J/\lambda$ and~$A_J$ is the numerical solution of~\eqref{eq:A} with~$r_{\rm in}=1.5$,~$r_{\rm out}=2$,~$\varepsilon=0.001$ and~$J=0.75 r_{\rm in}^2/(\sqrt8\lambda\varepsilon)$, and~$A_0,\,A_1$ are given by~\eqref{eq:approxAs} with the same parameters ({\color{blue} solid curve with `$\cdot$' markers at the data points}).  Super-imposed are the curve~$c_1\lambda^2$ where~$c_1\approx 2.19$ ({\color{red} dashes}), and~$c_2(\varepsilon/\lambda)^{1.6}$ where~$c_2\approx0.082$ ({\color{greenCurve} dash-dots}).}
\label{fig:errorAnalysis_HighCurrent}
\end{center}
\end{figure}
Result~\eqref{eq:psi_including_inner} shows how the width of the internal boundary layer of~$\psi$ depends on the current.  
This result also reveals the overall error introduced in the reduction from~\eqref{eq:system_ring_region} to~\eqref{eq:As_leadingOrder}.  Indeed, an~$O(\varepsilon/\lambda)$ error is introduced in the reduction from equation~\eqref{eq:A} to equation~\eqref{eq:As_full}.  An additional error of~$O(\lambda^2)$ is introduced by considering the solution form~\eqref{eq:As_expansion}.
Overall, the reduction from~\eqref{eq:system_ring_region} to~\eqref{eq:As_leadingOrder} introduced an error of~$O(\varepsilon/\lambda,\lambda^2)$.  
Therefore, one can expect an error of~$O(\varepsilon/\lambda,\lambda^2)$ in the approximation~$A_{\rm in}\approx A_0+\lambda A_1$.
A numerical verification of this result is presented in Figure~\ref{fig:errorAnalysis_HighCurrent}.  As expected, we observe that the error
\begin{equation}\label{eq:err_As}
E(\lambda)=\max|A_0(\rho)+\lambda A_1(\rho)-A_{\rm in}^{\rm numerical}(\rho;\lambda)|
\end{equation}
where~$A_{\rm in}^{\rm numerical}$ is computed numerically behaves as~$\lambda^2$ when~$\lambda^2\gg \varepsilon/\lambda$.  For smaller~$\lambda$, we observe that as expected the error decreases with~$\lambda$.  Particularly, we observe that the error is smaller than expected and behaves as~$(\varepsilon/\lambda)^c$ where~$c\approx 1.6$.  We also observe that in the region where both error terms are comparable, they cancel each other and further reduce the error, see graph in the region of~$\lambda\approx0.02$ in~Figure~\ref{fig:errorAnalysis_HighCurrent}.

Relation~\eqref{eq:A20} and the results~(\ref{eq:approxAs},\ref{eq:psi_including_inner}) imply that the above analysis is valid in the parameter regime~\eqref{eq:J_highcurrent_regime} and
\begin{equation}\label{eq:Js_bound}
J_s< \frac{r_{\rm in}^2}{\sqrt8}\quad \mbox{or}\quad J<\frac1{\varepsilon\lambda}\frac{r_{\rm in}^2}{\sqrt8}.
\end{equation}
Approximation~\eqref{eq:As_expansion} breaks down as~$J_s$ approaches this bound.  Indeed, in the limit~$J_s\longrightarrow r_{\rm in}^2/\sqrt8$, the denominator of~$A_1$ vanishes since~$\alpha^2-6\alpha+1\longrightarrow0$.  Consequently,~$\lambda A_1$ becomes dominant over~$A_0$, and the asymptotic expansion breaks down.  
Similarly, the correction term in~$\psi$ blows up in this limit.  Since~$r_{\rm in}^2/\sqrt8\approx 0.36 r_{\rm in}^2$, it is reasonable, in certain cases, to consider the region~$J_s\ll r_{\rm in}^2$.  In this case,~$A_0$, see~\eqref{eq:A0_result}, is approximated by
\begin{equation}\label{eq:A0_approx_small_Js}
A_0(\rho)=\frac{2}{r_{\rm in}^2}e^{-\rho}\left(1+\frac{3}{2}\delta^2- \frac12e^{-3\rho}\delta^2\right)+O(\delta^4),\quad \delta=\frac{J_s}{r_{\rm in}^2}.
\end{equation}
Particularly, the scaled solution~$\lambda\,A_{\rm in}$ agrees, up to~$O(\delta^2)$, with the profile~$B$~\eqref{eq:A_asympt} which arises in the low current case.  Figure~\ref{fig:generalPictureAJ_partialSC} demonstrates this point by comparing the numerical solution of~\eqref{eq:system_ring_region} for~$J_s=0.875$ where~$\frac{r_{\rm in}^2}{\sqrt8}=0.883$ (solid blue curve) with the profile~$B$ (dashed red curve).  As expected, we observe that in the region~$r>r_{\rm in}$ the two curves are indistinguishable except in a small region near~$r=r_{\rm in}$.

We next consider the case where the current is yet higher,~$J_s>\frac{r_{\rm in}^2}{\sqrt8}$, and~$1-\varepsilon^2 J^2A_J^2(r_{\rm in})<0$.

\subsection{Superconductivity destroyed in part of ring}\label{sec:SC_partially_destroyed}
Sections~\ref{sec:superconductive} and~\ref{sec:partialSC} considered the case of full or partial superconductivity.  Namely, the case in which the superconducting order parameter~$\psi$ remains strictly positive in any limit~$(\varepsilon,\lambda)\to0$ for which~$\varepsilon\ll\lambda$.  The analysis in these sections shows that full or partial superconductivity occurs for currents~$J$ below the treshold~\eqref{eq:Js_bound} for which~$1-\varepsilon^2J^2A_J^2(r)>0$ for~$r\in[r_{\rm in},r_{\rm out}]$.  
We now consider a case of higher currents, so that there exists a turning point~$r_{\rm in}<r_{\rm turning}<r_{\rm out}$ for which 
\begin{equation}\label{eq:def_r_transition}
\varepsilon J A_J(r_{\rm turning})=1.
\end{equation}
By Lemma~\ref{lem:AJ},~$A_J$ is strictly monotonically decreasing, hence 
\[
\begin{cases}
1-\varepsilon^2 J^2 A^2_J(r)<0,& r_{\rm in}\le r<r_{\rm turning},\\
1-\varepsilon^2 J^2 A^2_J(r)>0,& r_{\rm turning}\le r<r_{\rm out}.
\end{cases}
\]
This implies that, to leading order, see~\eqref{eq:psi_outer_cq},
\[
\psi=
\begin{cases}
0& r<r_{\rm turning},\\
\sqrt{1-\varepsilon^2J^2A_J^2}& r>r_{\rm turning}.
\end{cases}
\]
Namely, superconductivity is destroyed for~$r<r_{\rm turning}$, and the effective superconductive hollow cylinder is within the region~$r_{\rm turning}<r<r_{\rm out}$.

Let us first focus on the 
behavior of the solution near the turning point~$r_{\rm turning}$.  Note that, in what follows, we will show that the scaling~$A_J=O(\lambda)$ is not applicable for all~$r\in[r_{\rm in},r_{\rm out}]$.  Therefore, we do not use the scaling~\eqref{eq:As} as in the previous section, but rather consider~$A_J$.
A Taylor series expansion of~$A_J(r)$ about~$r=r_{\rm turning}$ yields for
\begin{equation}\label{eq:AJ_near_turning_point}
\begin{split}
\varepsilon^2 J^2A_J^2(r)&=\varepsilon^2 J^2A_J^2(r_{\rm turning})+2\varepsilon^2 J^2A_J(r_{\rm turning})A_J^\prime(r_{\rm turning})(r-r_{\rm turning})+O\left((r-r_{\rm turning})^2\right)\\&=1-\alpha_0\rhot+O(\lambda^2\rhot^2),\qquad \alpha_0=-2J_sA_J^\prime(r_{\rm turning}),\qquad \rhot=\frac{r-r_{\rm turning}}{\lambda},
\end{split}
\end{equation}
where the last equality is due to~\eqref{eq:Js} and~\eqref{eq:def_r_transition}.
Substituting~\eqref{eq:AJ_near_turning_point} in equation~\eqref{eq:psi} for~$\psi$, using the scaled variable~$\rho$ and neglecting~$O(\rhot^2,\varepsilon/\lambda)$ terms yields,
\begin{equation}\label{eq:nonlinearTurningPoint}
\frac{\varepsilon^2}{\lambda^2}\frac{d^2\psi}{d\rhot^2}
=\psi^3-\alpha_0\rhot\,\psi.
\end{equation}
This equation has the form of a Painlev\'{e} II equation.  Particularly, close enough to~$\rho=0$,~$\psi^3$ becomes dominant over~$\rho \psi$, and therefore it is not possible to neglect~$\psi^3$.  
Let
\begin{equation}\label{eq:psi_nu_relation}
\psi(\rho)=\sqrt2 \left[\frac{\alpha_0\,\varepsilon}{\lambda}\right]^{\frac13}\, \nu(y),\quad y=\left[{\frac{\alpha_0\lambda^2}{\varepsilon^2}}\right]^{\frac13}\rho.
\end{equation}
Then,~$\nu$ satisfies
\[
\nu^{\prime\prime}(y)+y \nu-2\nu^3=0.
\]
This equation admits multiple solutions which satisfy~$\lim_{y\to-\infty}\nu(y)=0$, but only one monotone solution which is the Hastings--McLeod solution~$\nu_0$ with the asymptotic behavior~\cite{hastings1980boundary}
\[
\nu_0=\begin{cases}
\sqrt {y/2}\left[1+\frac1{8y^3}+O\left(y^{-6}\right)\right],& y\gg1,\\
 \frac{1}{2\sqrt\pi}|y|^{-1/4}\exp\left(-\frac13|y|^{\frac32}\right),& y\ll -1.
\end{cases}
\]
The asymptotic behavior of the corresponding solution~$\psi$, see~\eqref{eq:psi_nu_relation}, takes the form
\[
\psi=\begin{cases}
\sqrt{\alpha_0\,\rho}
,& \left[\frac{\varepsilon^2}{\alpha_0\lambda^2}\right]^{\frac13}\ll \rho\ll1,\\
\sqrt{\frac{2}{\pi}}\sqrt{\frac{\varepsilon}{\lambda}}\left|\frac{\alpha_0}{\rho}\right|^{\frac14}\exp\left(-\frac{\sqrt2\,\alpha_0^{\frac56}}3\left|\frac{\lambda}{\varepsilon}\right|^{\frac23}\rho^{\frac32}\right),& \rho\ll-\left[\frac{\varepsilon^2}{\alpha_0\lambda^2}\right]^{\frac13}.\end{cases}
\]
The intermediate region~$\left[\frac{\varepsilon^2}{\alpha_0\lambda^2}\right]^{\frac13}\ll \rho\ll1$ is also the matching region between the Hastings-McLeod solution and the leading order approximation of~$\psi$.  Indeed, in this region, see~\eqref{eq:AJ_near_turning_point},
\[
\psi=\sqrt{1-\varepsilon^2 J^2A_J^2(r)}\approx \sqrt{\alpha_0\,\rho}.
\]
Overall,
\begin{equation}\label{eq:g_approximation}
\psi^{\rm approx}\approx\begin{cases}
\sqrt2 \left[\frac{\alpha_0\,\varepsilon}{\lambda}\right]^{\frac13} \nu\left(\sqrt[3]{\frac{\alpha_0\lambda^2}{\varepsilon^2}}\rho\right), &  \rho \ll 1,\\
\sqrt{1-\varepsilon^2 J^2A_J^2(r)},& \left[\frac{\varepsilon^2}{\alpha_0\lambda^2}\right]^{\frac13}\ll \rho.
\end{cases}
\end{equation}
Figure~\ref{fig:psiProfile} presents the numerical solution~$\psi$ of~\eqref{eq:psi} near the turning point, see solid blue curve, as well as the Hastings-McLeod solution~$\nu$ scaled according to~\eqref{eq:g_approximation} where~$\alpha_0$ is computed numerically, see red dashed curve, and the curve~$\sqrt{1-\varepsilon^2 J^2A_J^2(r)}$, see green dash-dotted curve.  As expected, for small~$\rho$ the numerical solution~$\psi$ agrees with the Hastings-McLeod solution, while for larger positive~$\rho$ it agrees with the function~$\sqrt{1-\varepsilon^2 J^2A_J^2(r)}$.  The matching region in which the solution agrees with both the above approximations is a small region~$\rho>\rho_{\rm match}$ where~$\rho_{\rm match}=\left[\frac{\varepsilon^2}{\alpha_0\lambda^2}\right]^{\frac13}$. 

\begin{figure}[ht!]
\begin{center}
\scalebox{0.75}{\includegraphics{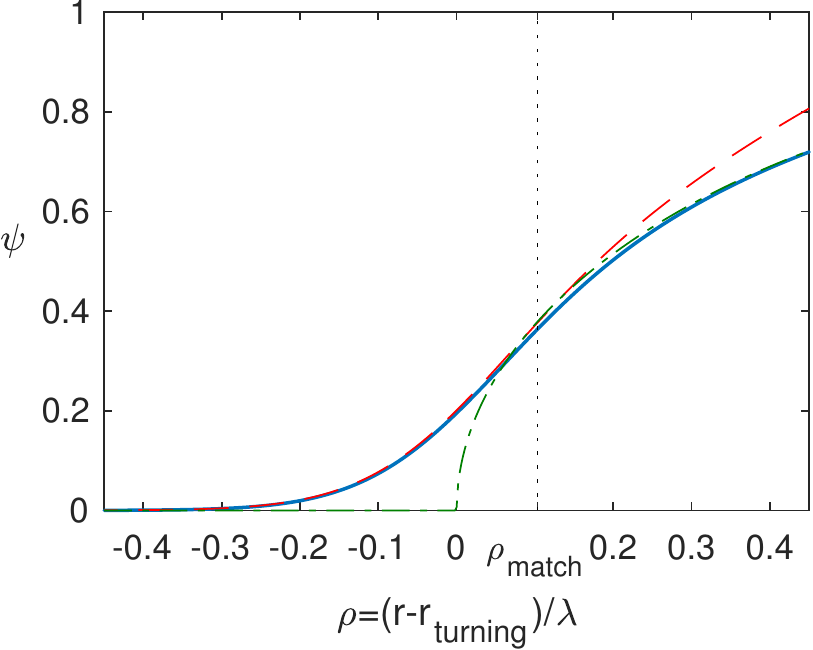}}
\caption{Solution~$\psi$ of~\eqref{eq:psi} ({\color{blue} solid}) 
for~$r_{\rm in}=0.5$,~$r_{\rm out}=1$,~$\varepsilon=0.001$,~$\lambda=0.025$ and~$J=8703.4$ ({\color{blue} solid}).  Super-imposed is the Hastings-McLeod solution~$\nu$ scaled according to~\eqref{eq:g_approximation} where~$\alpha_0$ is computed numerically ({\color{red} dashes}), the function~$\sqrt{1-J_sA_{\rm in}^2}$ in the region~$1-J_sA_{\rm in}^2\ge0$ ({\color{greenCurve} dash-dots}), and the curve~$\rho_{\rm match}=\left[\frac{\varepsilon^2}{\alpha_0\lambda^2}\right]^{\frac13}$ (dots).}
\label{fig:psiProfile}
\end{center}
\end{figure}

The result~\eqref{eq:g_approximation} shows that~$\psi$ undergoes a transition from~$0$ to~$\sqrt{1-\varepsilon^2 J^2A_J^2(r)}$ in a narrow transition layer whose width scales as~$\lambda^{1/3}\varepsilon^{2/3}$.  This result, however, does not reveal the location of the transition layer, i.e., the value of~$r_{\rm turning}$.
A classic way to determine~$r_{\rm turning}$ is to fully address the underlying turning point problem - namely approximate the solution in the transition layer, while matching it to the solutions in the left and right domains, see, e.g.,~\cite{karasev2003global,bender2013advanced}.  This is a demanding asymptotic analysis problem for several reasons: The location of the turning point~$r_{\rm turning}$ is not apriori known, the turning point problem in~$\psi$ is nonlinear, see~\eqref{eq:nonlinearTurningPoint}, and it involves a system where one function,~$A_J$, decays exponentially towards the boundary at~$r_{\rm out}$, while the second function,~$\psi$, decays exponentially in the opposite direction, towards the boundary at~$r_{\rm in}$.  The latter implies that the choice of the matching direction is non-trivial.  In this paper, we apply a much simpler approach and use a variational approximation.  Particularly, we exploit the fact that the relevant solution of~\eqref{eq:system} is a (local) minimizer of the corresponding free energy functional~\eqref{eq:energyAsc}, and find~$r_{\rm turning}$ for which an appropriate approximation of the solution will have minimal energy.  As will be shown, roughly speaking, the solution in the transition layer does not contribute much to the overall energy.  Therefore, a key advantage of this approach is that it does not require resolving to high accuracy the solution in the transition layer.

To apply a variational approximation, let us consider a solution of the form~\eqref{eq:g_approximation} in the limit of zero width of a transition layer,
\begin{subequations}\label{eq:variationalApprox}
\begin{equation}\label{eq:psiApprox}
\psi=\begin{cases}
0,& r<r_{\rm eff},\quad r>r_{\rm out},\\
\sqrt{1-\varepsilon^2 J^2A_J^2},& r_{\rm eff}<r<r_{\rm out},\\
\end{cases}
\end{equation}
where~$r_{\rm eff}$ is the, apriori unknown, effective location of the transition layer.  Note that while~$r_{\rm eff}$ is closely related to the turning point~$r_{\rm turning}$ defined by~\eqref{eq:def_r_transition}, it is not necessarily equal to it.

In the region~$r<r_{\rm eff}$, Equation~\eqref{eq:A} is homogenous and has an explicit solution of the form~$A_J=c r+1/r$ which satisfies
\begin{equation}\label{eq:BC_at_reff}
A^\prime_J(r_{\rm eff})-\frac{A_J}{r_{\rm eff}}=-\frac2{r_{\rm eff}^2}.
\end{equation}
In the region~$r>r_{\rm eff}$,~$A_J$ satisfies, to leading order, equation~\eqref{eq:A_nonlinear} for~$r_{\rm eff}<r<r_{\rm out}$.  This implies that one may use approximation~\eqref{eq:approxAs} of~$A_J$ in this region as a basis for the variational approximation of~$r_{\rm eff}$.  The resulting expression, however, involve an implicit equation for~$r_{\rm eff}$ which includes integrals that cannot be explicitly resolved.  Note, however, that outside the transition layer, approximation~\eqref{eq:approxAs} is, to leading order, of the form~$c\,e^{-\rho}$, see~\eqref{eq:A0_approx_small_Js}.  This form coincides with the solution form in the low flux regime, see~\eqref{eq:leadingOrderSmallJ}.  
Motivated by this fact, we choose~$r_{\rm eff}$
so that
\[
A_J(r)=\frac{2\lambda}{r_{\rm eff}^2}e^{-\rho},\quad \rho=\frac{r-r_{\rm eff}}{\lambda}.
\]
Namely, we seek for an effective superconductive inner ring radius,~$r_{\rm eff}$, so that for~$r>r_{\rm eff}$ the solution agrees, to leading order, with the scaled solution~$A_J$ in the low flux regime~\eqref{eq:low_current_regime}.  The resulting ansatz takes the form
\begin{equation}\label{eq:AJapprox}
A_J^{\rm approx}(r;r_{\rm eff})= \begin{cases}
\frac1r-c\,r,& r\le r_{\rm eff}\\
\frac{2\lambda}{r_{\rm eff}^2}e^{-\rho} ,& r>r_{\rm eff}\\
\end{cases},\quad c=\frac{r_{\rm eff}-2\lambda}{r_{\rm eff}^3},
\end{equation}
\end{subequations}
where~$c$ is chosen to assure continuity at~$r=r_{\rm eff}$.  

Substituting the approximate quantities~\eqref{eq:variationalApprox} into the corresponding free energy~\eqref{eq:energyAsc} yields
\begin{equation*}
\begin{split}
\Energy(r_{\rm eff})=&\frac{2(r_{\rm eff} - 2\lambda)^2J_s^2}{r_{\rm eff}^4}+\frac{\lambda J_s^2}{r_{\rm eff}^4}\left[2 r_{\rm eff}+4 \lambda e^{\frac{2r_{\rm eff}}{\lambda}}E_1\left(\frac{2r_{\rm eff}}{\lambda}\right)- 3\lambda\right]+\frac{2\lambda J_s^2}{r_{\rm eff}^3}+\\& - \frac{J_s^2(4r_{\rm eff} - 1)\lambda^2}{r_{\rm eff}^4}+\frac{r_{\rm eff}^2}4+\frac{2J_s^4}{r_{\rm eff}^7} \lambda+\frac{J_s^4-16 J_s^4 r_{\rm eff}}{2r_{\rm eff}^8} \lambda^2+O(\lambda^3,e^{\frac{r_{\rm eff}-r_{\rm out}}{\lambda}},\varepsilon^2),
\end{split}
\end{equation*}
where~$E_1(x)=\int_1^\infty \frac1te^{-t x}dt$ is the exponential integral.
Using first order condition~$\Energy^\prime(r_{\rm eff})=0$ to compute the value of~$r_{\rm eff}$ for which the energy is minimal yields
\begin{equation}\label{eq:reff_approx}
r_{\rm eff}=\sqrt[4]{8}\sqrt{J_s}-\frac{55}{64}\lambda+O(\lambda^2).
\end{equation}
The value~$r_{\rm eff}$ is the effective inner ring radius of the superconductive ring in the sense that for~$r<r_{\rm eff}$, the order parameter~$\psi\approx0$.

The effective inner cylinder radius must reside in the cylinder region,~$r_{\rm in}\le r_{\rm eff}\le r_{\rm out}$.
Substituting~\eqref{eq:reff_approx} in the above bound for~$r_{\rm eff}$ and isolating~$J_s$ yields, to leading order, the flux regime 
\begin{equation}\label{eq:highCurrentRegime_destory_SC}
\frac{r_{\rm in}^2}{\sqrt8}<J_s<\frac{r_{\rm out}^2}{\sqrt8}.
\end{equation}
This flux regime complements the regime~\eqref{eq:Js_bound} in which the whole cylinder region is fully or partially superconductive. 

\begin{figure}[ht!]
\begin{center}
\scalebox{0.5}{\includegraphics{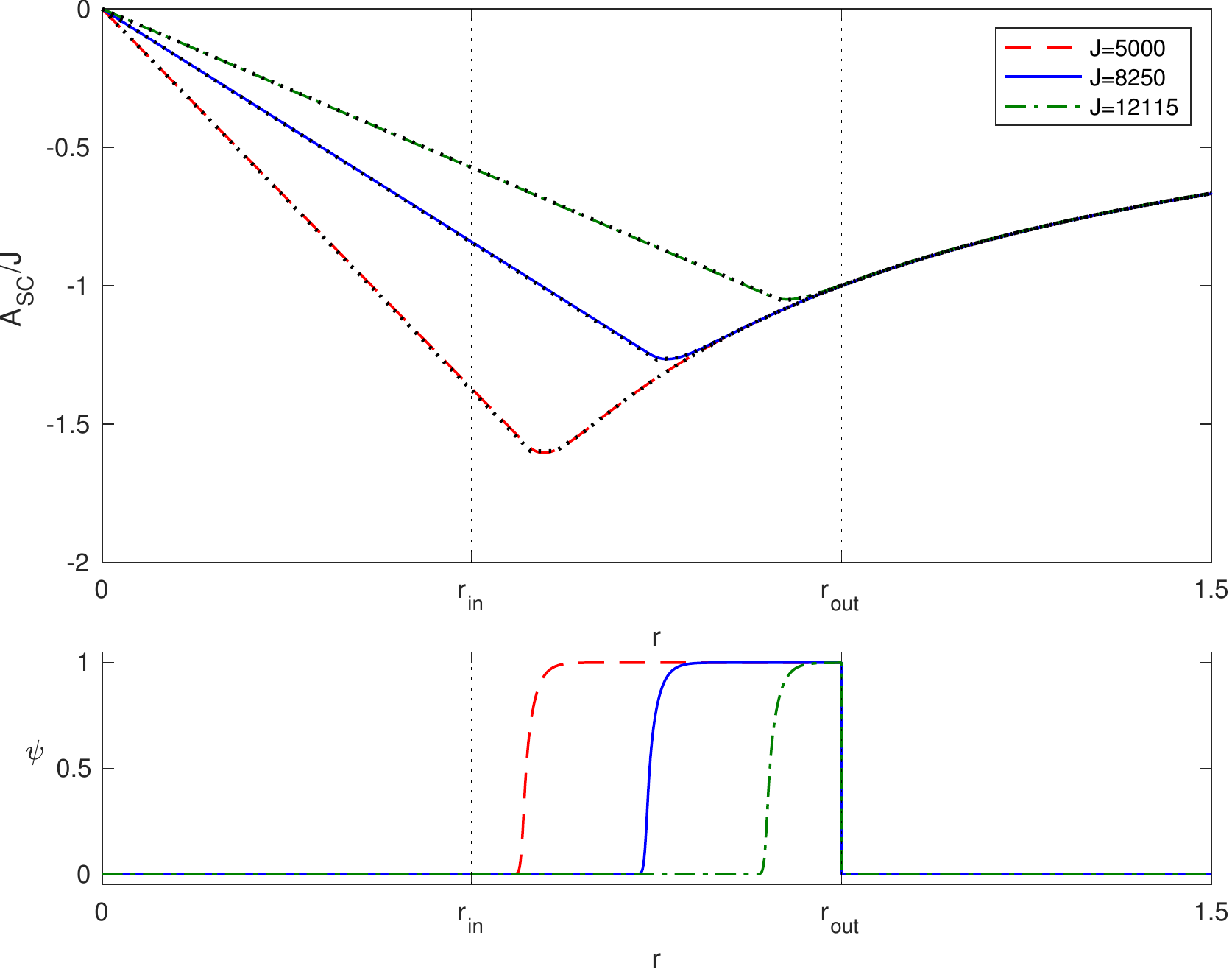}}
\caption{A: Graph of scaled vector potential~$A_{\rm sc}/J=A_J-1/r$, where~$A_J$ is the numerical solution of system~\eqref{eq:system_ring_region} for~$r_{\rm in}=0.5$,~$r_{\rm out}=1$,~$\varepsilon=0.001$, and~$\lambda=0.025$ for~$J=5000$ ({\color{red} dashes}),~$J=8250$ ({\color{blue} solid}), and~$J=12115$ ({\color{greenCurve} dash-dots}).  For each curve, super-imposed is the corresponding approximation~(\ref{eq:AJapprox},\ref{eq:reff_approx}) (dotted curves).  Each pair of curves are indistinguishable.  B: Graph of order parameter~$\psi$ corresponding to the three solutions of~\eqref{eq:system_ring_region} presented in A.}
\label{fig:generalPicture_partialSC}
\end{center}
\end{figure}The above analysis and its result~\eqref{eq:reff_approx} rely on the ansatz~\eqref{eq:variationalApprox}.    
Figure~\ref{fig:generalPicture_partialSC} presents several solutions of~\eqref{eq:system_ring_region} in the flux regime~\eqref{eq:highCurrentRegime_destory_SC}, along with their corresponding approximations~\eqref{eq:reff_approx}.  As expected, the effective inner cylinder radius increases with the current such that superconductivity is destroyed in the inner part of the cylinder, compare with Figure~\ref{fig:generalPictureAJ_lowCurrent} in the low flux regime.  In all examples, the solutions and their approximations are indistinguishable, implying that ansatz~\eqref{eq:variationalApprox} successfully describes the solution behavior, at least to leading order.  To better quantify the accuracy of approximation~\eqref{eq:reff_approx}, in Figure~\ref{fig:error_reff} we present the approximation error~$E=|r_{\rm eff}-r_{\rm eff}^{\rm numerical}|$ where~$r_{\rm eff}^{\rm numerical}$ is defined as the point in which  
\begin{equation}\label{eq:def_reff_numerical}
A_J(r_{\rm eff}^{\rm numerical})=\frac{2\lambda}{[r_{\rm eff}^{\rm numerical}]^2}.
\end{equation}
\begin{figure}[ht!]
\begin{center}
\scalebox{0.5}{\includegraphics{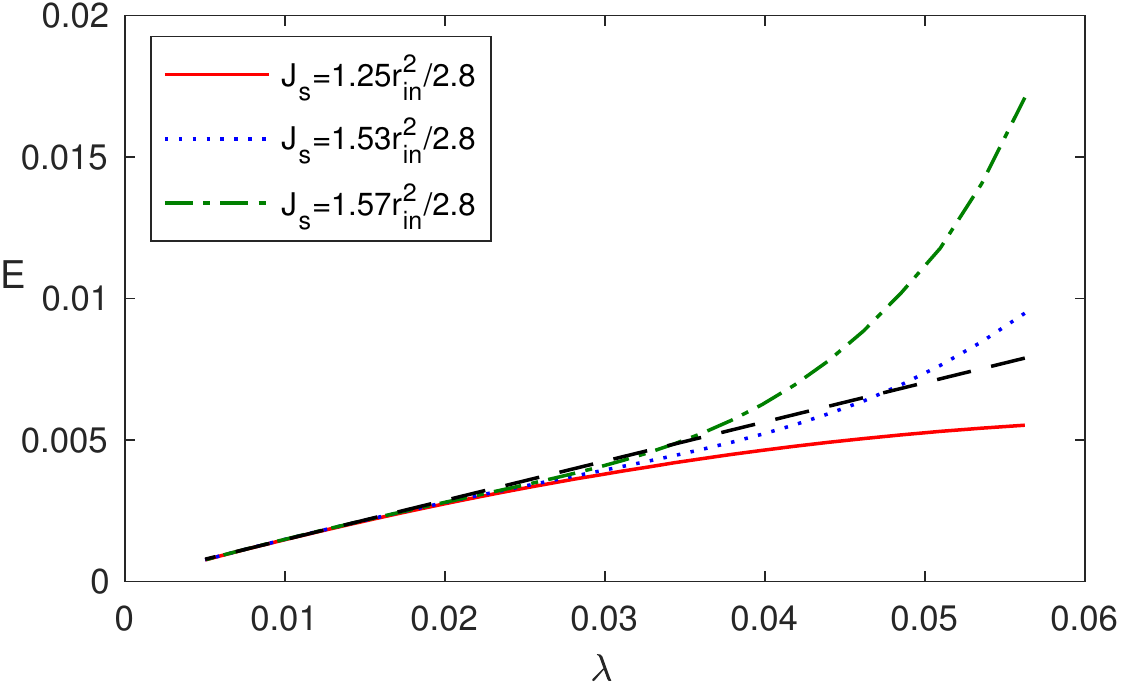}}
\caption{Error~$E=|r_{\rm eff}-r_{\rm eff}^{\rm numerical}|$ where~$r_{\rm eff}^{\rm numerical}$ is defined by~\eqref{eq:def_reff_numerical} where~$A_J$ is the numerical solution of system~\eqref{eq:system_ring_region} for~$r_{\rm in}=1.5$,~$r_{\rm out}=2$,~$\varepsilon=0.001$ and scaled current values~$J_s=1.25r_{\rm in}^2/\sqrt8$ ({\color{red} solid}),~$J_s=1.53r_{\rm in}^2/\sqrt8$ ({\color{blue} dots}) and~$J_s=1.57r_{\rm in}^2/\sqrt8$ ({\color{greenCurve} dash-dots}). Superimposed is the fitted curve~$a+\lambda b$ where~$a\approx 9.66\cdot 10^{-5}$ and~$b\approx 0.138$ (dashes).}
\label{fig:error_reff}
\end{center}
\end{figure}
We observe that~$E$ decays linearly with~$\lambda$.  We observe that the source of this error is the~$O(\lambda)$ accuracy of the derivative of the ansatz~\eqref{eq:AJapprox} \[
A_J^{\rm approx}(r;r_{\rm eff})-A_J(r;r_{\rm eff})=O(\lambda^2),\quad \frac{d}{dr}\left[A_J^{\rm approx}(r;r_{\rm eff})-A_J(r;r_{\rm eff})\right]=O(\lambda).
\]

\section{The weakly superconductive case~$\psi\ll1$}\label{sec:weakSC}
Section~\ref{sec:highCurrentRegime} considered the high flux regime, up to~$J\approx \frac{r_{\rm out}^2}{\sqrt8\varepsilon\lambda}$.  In this high flux regime, we have shown that the superconducting state order parameter~$\psi$ decreases with~$J$.  Particularly, as current increases, superconductivity is gradually reduced and eventually destroyed in increasing parts of the ring.   Based on this emerging picture, it is reasonable to assume that upon further increase of the current~$J$, superconductivity will gradually decrease until it is completely destroyed.  Indeed, it is well known that superconductivity is destroyed at sufficiently large currents~\cite{tinkham2004introduction,schrieffer2018theory}.  In this section, we focus on the regime of weak superconductivity~$\psi\ll1$.
Significantly, we do not make any assumptions on the magnitude of the current~$J$.  Rather, we aim to obtain from the analysis an approximation of the current~$J$ at which superconductivity is destroyed everywhere in the cylinder.  

Let us consider a solution of~\eqref{eq:system_ring_region} for which~$\psi\ll1$.   The equation form of~\eqref{eq:psi} and the analysis conducted in the previous section, strongly suggest that outside a layer of~$O(\varepsilon)$ from the boundaries and to leading order,
\begin{equation}\label{eq:psi_leadingOrder}
\psi=\begin{cases}
0,& 1-\varepsilon^2 J^2A_J^2<0,\\
\sqrt{1-\varepsilon^2 J^2A_J^2},& 1-\varepsilon^2 J^2A_J^2>0.\\
\end{cases}
\end{equation}
In this case,~$\psi\ll1$ implies that~$1-\varepsilon^2 J^2A_J^2\ll1$, and in particular,~$A_J(r_{\rm out})\approx \frac{1}{\varepsilon J}$ or~$A_J(r_{\rm out})\approx A_J(r_{\rm turning})$, where~$A_J(r_{\rm turning})=\frac1{\varepsilon J}$, see~\eqref{eq:def_r_transition}.  

The emerging picture is that the case~$\psi\ll1$ corresponds to the case where the turning point is very close to the outer rim of the cylinder
\[
r_{\rm turning}\approx r_{\rm out}.
\]  
Note that the turning point,~$r_{\rm turning}$, is closely related to the effective inner cylinder radius from which the cylinder is superconductive, see Section~\ref{sec:SC_partially_destroyed}.  Therefore, as expected, when~$\psi\ll1$, superconductivity is destroyed in nearly all the cylinder region, except in some narrow region near~$r=r_{\rm out}$.  

Let us consider the function~$A_J$ in the region near~$r=r_{\rm out}$
\[
A_{\rm in}\left(\rho\right)=\varepsilon J A_J,\quad \rho=\frac{r_{\rm out}-r}{\lambda}.
\]
In terms of the scaled variables, and in the region~$r_{\rm turning}<r=r_{\rm out}$ in which~$\psi\approx \sqrt{1-\varepsilon^2 J^2A_J^2}$, Equation~\eqref{eq:A} takes the form
\begin{subequations}\label{eq:weakSC}
\begin{equation}\label{eq:A_weakSC}
A_{\rm in}^{\prime\prime}(\rho)-\frac{\lambda\,A_{\rm in}^\prime}{r_{\rm out}-\lambda \rho}-\frac{\lambda^2\,A_{\rm in}}{(r_{\rm out}-\lambda \rho)^2}=A_{\rm in}-A_{\rm in}^3,\quad -A^\prime_i(0)+\frac{\lambda A_{\rm in}(0)}{r_{\rm out}}=0.
\end{equation}
Equation~\eqref{eq:A_weakSC} is a second order equation, and therefore an additional boundary condition is required.
The assumption~$\psi\approx \sqrt{1-\varepsilon^2 J^2A_J^2}$ is valid only beyond the transition layer around~$r=r_{\rm turning}$, see~\eqref{eq:g_approximation}.  As will be shown, 
\[
\rho_{\rm turning}=\frac{r_{\rm out}-r_{\rm turning}}{\lambda}=O(1).
\]
Therefore, one cannot use classic matching condition at~$\rho\gg1$ to obtain a second boundary condition in~\eqref{eq:A_weakSC}.  
Similar to the derivation of~\eqref{eq:variationalApprox}, in the limit of a transition layer with zero width,~$A_J$ takes the form
\begin{equation}\label{eq:AJform_weakSC}
A_J=\frac{1}{r}-c\,r,\quad c=\frac{\varepsilon J - r_{\rm turning}}{\varepsilon J\,r_{\rm turning}^2},\qquad r\le r_{\rm turning},
\end{equation}
where~$c$ is set such that~$\varepsilon J A_J(r_{\rm turning})=1$.
Using definition~\eqref{eq:AJform_weakSC} and~\eqref{eq:def_r_transition} gives rise to the approximate boundary condition at~$\rho_{\rm turning}$,
\begin{equation}\label{eq:A_weakSC_BC}
A_{\rm in}(\rho_{\rm turning})=1,\quad A^\prime_{\rm in}(\rho_{\rm turning})=\lambda\frac{2\varepsilon J-r_{\rm turning}}{ r_{\rm turning}^2}.
\end{equation}
\end{subequations}
The resulting problem~\eqref{eq:weakSC} is a free boundary problem, in which the location of the boundary at~$\rho=\rho_{\rm turning}$ is apriori unknown.    Accordingly, it has three boundary conditions.  Let us seek for a solution of~\eqref{eq:weakSC} for~$\lambda\ll1$ in the form
\begin{equation}\label{eq:As_weak_expansion}
A_{\rm in}(\rho)=A_0(\rho)+\lambda A_1(\rho)+O(\lambda^2),
\end{equation}
Substituting~\eqref{eq:As_weak_expansion} in~\eqref{eq:weakSC} and equating the~$O(1)$ and~$O(\lambda)$ terms yields 
\begin{equation}\label{eq:As_weak_leadingOrder}
A_0^{\prime\prime}(\rho)-A_0=-A_0^3,\qquad A_0^\prime(0)=0,\quad A_0(\rho_{\rm turning})=1,\quad A_0^\prime(\rho_{\rm turning})=0
\end{equation}
and
\begin{equation}\label{eq:As_weak_firstOrder}
A_1^{\prime\prime}(\rho)-A_1=\frac{A_0^\prime}{r_{\rm out}}-3A_0^2A_1,\qquad A_1^\prime(0)=\frac{A_0(0)}{r_{\rm out}},\quad A_1(\rho_{\rm turning})=0,\quad A_1^\prime(\rho_{\rm turning})=\frac{2\varepsilon J- r_{\rm turning}}{r_{\rm turning}^2},
\end{equation}
respectively.

Equation~\eqref{eq:As_weak_leadingOrder} yields~$A_0\equiv1$.  Substituting this result in~\eqref{eq:As_weak_firstOrder} and using the boundary condition~$A_1(\rho_{\rm turning})=0$ yields
\[
A_1=c\,\sin(\sqrt{2}(\rho-\rho_{\rm turning})),\qquad A_1^\prime(0)=\frac{1}{r_{\rm out}},\quad A_1^\prime(\rho_{\rm turning})=\frac{2\varepsilon J- r_{\rm turning}}{r_{\rm turning}^2}.
\]
The solution form implies that~$\rho_{\rm turning}=O(1)$, otherwise~$A_1$ is oscillatory and~$A_J$ is not a monotone function, in contrast to Lemma~\ref{lem:AJ}.  Therefore,~$r_{\rm turning}=r_{\rm out}+O(\lambda)$.  This result enables using the boundary condition for~$O(\lambda)$ terms of~\eqref{eq:As_weak_expansion},
\begin{equation}\label{eq:A1_BC_turning_with_rout}
A_1^\prime(\rho_{\rm turning})=\frac{2\varepsilon J- r_{\rm out}}{r_{\rm out}^2},
\end{equation}
which implies that
\[
c=\frac{2\varepsilon J- r_{\rm out}}{\sqrt{2}r_{\rm out}^2}.
\]
The boundary condition at~$\rho=0$, see~\eqref{eq:As_weak_firstOrder}, implies that 
\begin{equation}\label{eq:rho_turning}
\cos(\sqrt2\rho_{\rm turning})=\frac{r_{\rm out}}{2\varepsilon J- r_{\rm out}}
\end{equation}
Overall, in the region~$r_{\rm turning}<r<r_{\rm out}$,~$A_J$ is approximated by
\begin{equation}\label{eq:Ai_result}
A_J^{\rm approx}=\frac{1}{\varepsilon J}\left[1-\lambda\frac{2\varepsilon J- r_{\rm out}}{\sqrt{2}r_{\rm out}^2} \sin\left(\sqrt{2}(\rho_{\rm turning}-\rho)\right)\right],\quad \rho=\frac{r_{\rm out}-r}{\lambda},\qquad r_{\rm turning}<r<r_{\rm out},
\end{equation}
where~$\rho_{\rm turning}$ satisfies~\eqref{eq:rho_turning}.  

Figure~\ref{fig:weakSC_error} presents the approximation error
\begin{equation}\label{eq:E_weakSC}
E=\max_{r_{\rm turning}<r<r_{\rm out}}|A_J-A_J^{\rm approx}|,
\end{equation}
where~$A_J$ is the numerical solution of system~\eqref{eq:system_ring_region}, and~$A_J^{\rm approx}$ is the corresponding approximation~\eqref{eq:Ai_result}. We observe that for~$J=6.5/\varepsilon$, the error decreases linearly,~$E\sim 0.15\lambda+8\cdot 10^{-4}$, and that the overall error decreases with increasing~$J$.  Similar to the derivation of~\eqref{eq:variationalApprox}, the error is~$O(\lambda)$ although expansion~\eqref{eq:As_weak_expansion} is up to~$O(\lambda^2)$, due to the error introduced by the approximate boundary conditions~\eqref{eq:A_weakSC_BC}.  We observe that the error is plotted in a region in which it attains a minimum.  A likely cause for this behavior is a high order error component of the form~$\varepsilon^\gamma/\lambda$ as in Figure~\ref{fig:errorAnalysis_HighCurrent}, although full resolution of high order terms is required to verify this claim.  
\begin{figure}[ht!]
\begin{center}
\scalebox{0.6}{\includegraphics{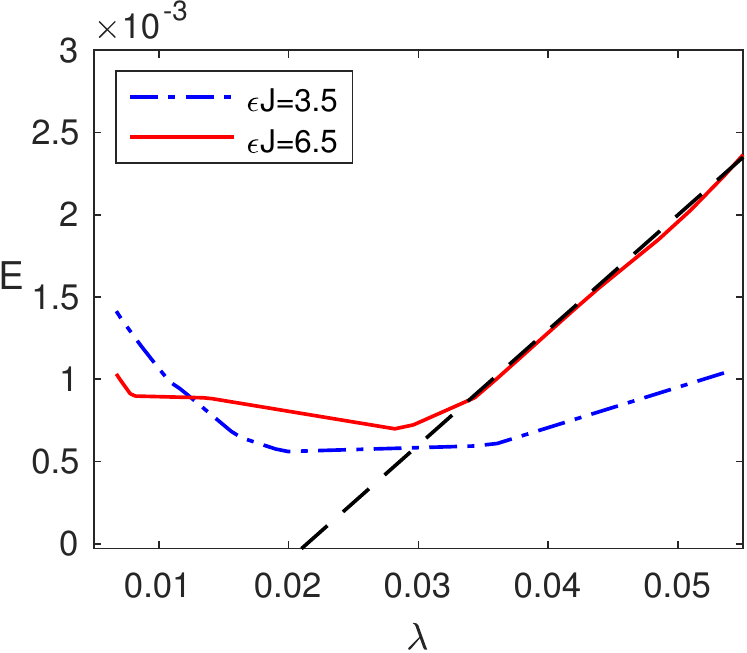}}
\caption{Error~$E$ given by~\eqref{eq:E_weakSC} where~$A_J$ is the numerical solution of system~\eqref{eq:system_ring_region} for~$r_{\rm in}=1.5$,~$r_{\rm out}=2$,~$\varepsilon=0.001$ and current values~$J=6.5/\varepsilon$ ({\color{red} solid}) and~$J=3.5/\varepsilon$ ({\color{blue} dash dots}). Superimposed is the fitted curve~$a+\lambda b$ where~$a\approx -1.5\cdot 10^{-3}$ and~$b\approx 0.07$ (dashes).}
\label{fig:weakSC_error}
\end{center}
\end{figure}

\begin{figure}[ht!]
\begin{center}
\scalebox{0.5}{\includegraphics{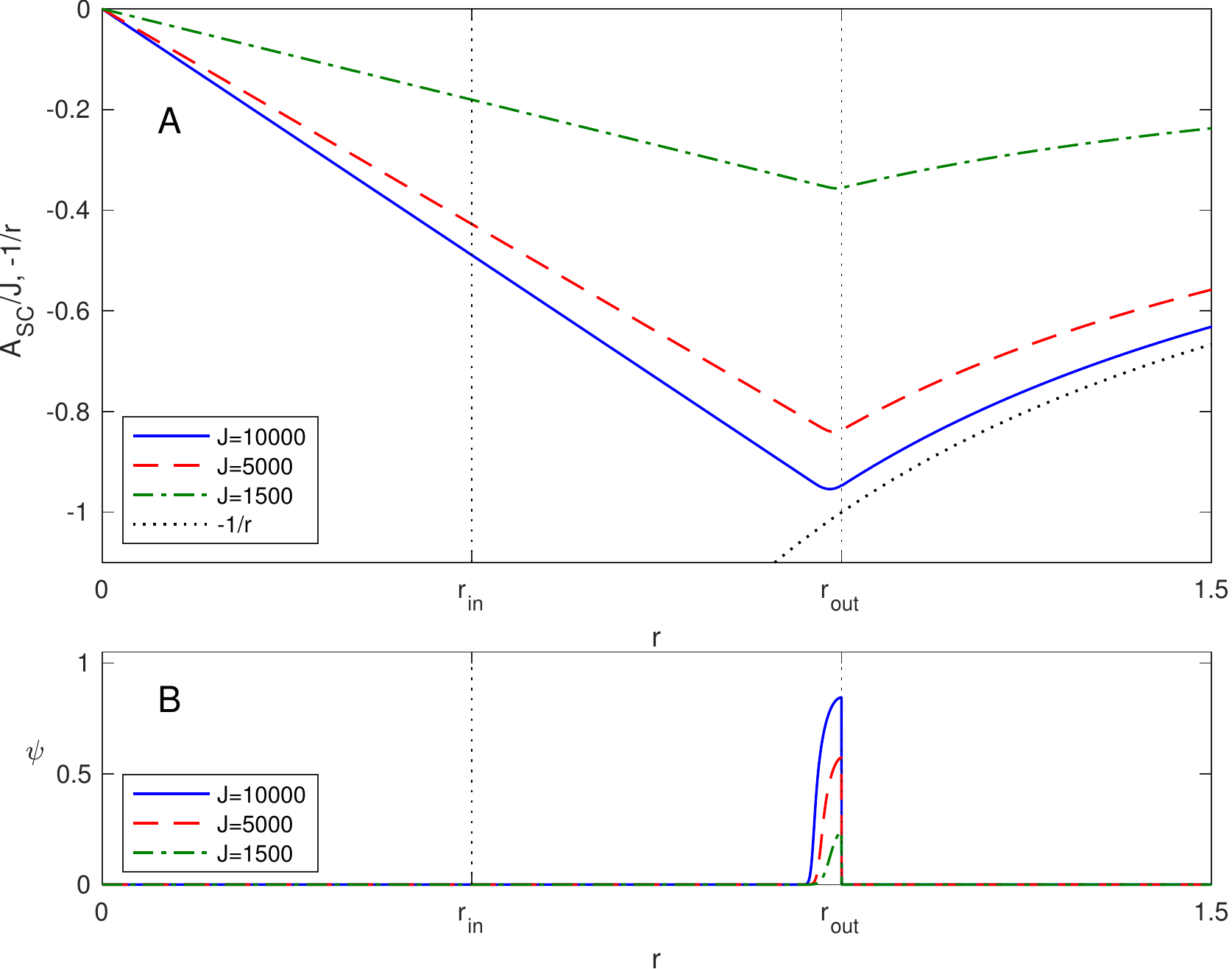}}
\caption{A: Graph of scaled vector potential~$A_{\rm sc}/J=A_J-1/r$, where~$A_J$ is the numerical solution of system~\eqref{eq:system_ring_region} for~$r_{\rm in}=0.5$,~$r_{\rm out}=1$,~$\varepsilon=0.001$, and~$\lambda=0.025$ for~$J=10000$ ({\color{blue} solid}),~$J=5000$ ({\color{red} dashes}), and~$J=1500$ ({\color{greenCurve} dash-dots}).  B: Graph of order parameter~$\psi$ corresponding to the three solutions of~\eqref{eq:system_ring_region} presented in A.}
\label{fig:generalPictureAJ_weakSC}
\end{center}
\end{figure}In contrast to the solution characterized in Sections~\ref{sec:superconductive} and~\ref{sec:highCurrentRegime} for which 
\[
A_J(r_{\rm out})\approx 0,
\]
the solutions approximated by~\eqref{eq:Ai_result} satisfy
\begin{equation}\label{eq:Ai_result_rout}
A_J^{\rm approx}(r_{\rm out})=\frac{-\lambda \sqrt2\sqrt{\varepsilon J}\sqrt{\varepsilon J - r_{\rm out}} + r_{\rm out}^2}{\varepsilon J r_{\rm out}^2}
\end{equation}
This implies that outside the cylinder~$r>r_{\rm out}$, the scaled vector potential~$A_{\rm sc}/J$, where~$A_{\rm sc}$ is defined by~\eqref{eq:rel_cylinder_domain_to_R}, does not equal to~$-1/r$, see Figure~\ref{fig:generalPictureAJ_weakSC} in comparison with Figures~\ref{fig:generalPictureAJ_lowCurrent} and~\ref{fig:generalPicture_partialSC}.

The result~\eqref{eq:Ai_result} is valid when~$\lambda A_1\ll A_0$ or~$\varepsilon\lambda J\ll r_{\rm out}^2$.  Furthermore, 
expression~\eqref{eq:Ai_result_rout} implies that~$\varepsilon J>r_{\rm out}$.  Overall,
\begin{equation}\label{eq:weakSC_currentregime}
\frac{r_{\rm out}}{\varepsilon }<J\ll \frac{r_{\rm out}^2}{\varepsilon\lambda}.
\end{equation}
This flux regime overlaps with the flux regimes studied in Sections~\ref{sec:superconductive} and~\ref{sec:highCurrentRegime}.  The implication is that at the overlap region the system~\eqref{eq:system_ring_region} has multiple non-trivial solutions.  Indeed, two different solutions of~\eqref{eq:system_ring_region} with the same parameters and, in particular,~$J=1500$ are presented in Figure~\ref{fig:generalPictureAJ_weakSC} and Figure~\ref{fig:generalPictureAJ_lowCurrent}.  Similarly, two different solutions of~\eqref{eq:system_ring_region} with the same parameters and, in particular,~$J=5000$ are presented in Figure~\ref{fig:generalPictureAJ_weakSC} and Figure~\ref{fig:generalPicture_partialSC}.
\begin{figure}[ht!]
\begin{center}
\scalebox{0.6}{\includegraphics{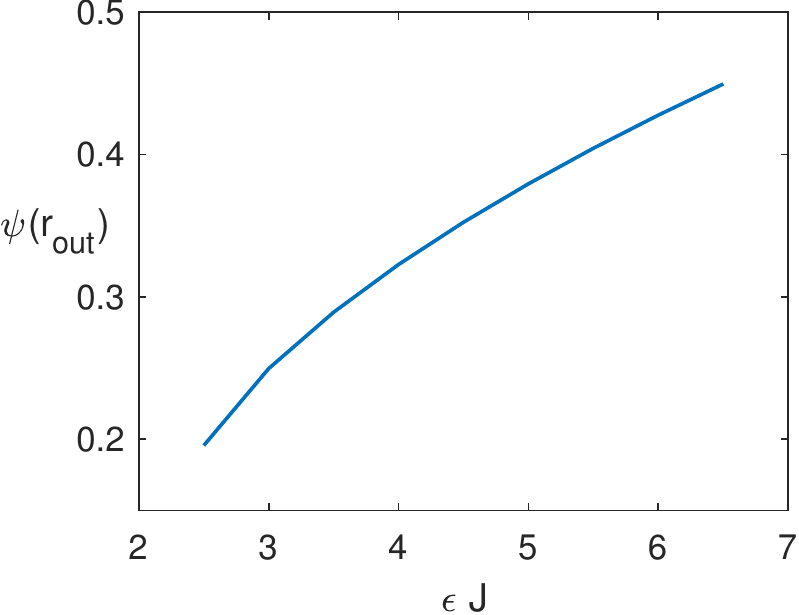}}
\caption{Value~$\psi(r_{\rm out})$ as function of~$J$ where~$\psi$ is the numerical solution of system~\eqref{eq:system_ring_region} for~$r_{\rm in}=1.5$,~$r_{\rm out}=2$,~$\varepsilon=0.001$ and~$\lambda=0.05$.}
\label{fig:weakSC_psiM}
\end{center}
\end{figure}

Furthermore, surprisingly, we observe in Figure~\ref{fig:generalPictureAJ_weakSC} that the superconducting state order parameter~$\psi$ increases with the current~$J$, see also Figure~\ref{fig:weakSC_psiM}.  Indeed, substituting~\eqref{eq:As_weak_expansion} in~\eqref{eq:psi_leadingOrder}, yields that
\[
\frac{d\,\psi(r_{\rm out})}{dJ}>0.
\]
Therefore, while the analysis in Sections~\ref{sec:superconductive} and~\ref{sec:highCurrentRegime} suggested that upon an increase of current, superconductivity will gradually decrease as until it is completely destroyed, the above results give rise to a very different picture.  Indeed, we find that superconductivity could be destroyed in a regime~\eqref{eq:weakSC_currentregime} of currents  which are significantly smaller than the currents in the `high flux regime'~\eqref{eq:J_highcurrent_regime} studied in Section~\ref{sec:highCurrentRegime}.  

\section{Numerical continuation study}\label{sec:numericalContinuation}
The analysis presented in Sections~\ref{sec:superconductive} and~\ref{sec:highCurrentRegime} shows that at low applied flux, the cylinder region is fully superconductive.  Then, as the flux~$J$ increases to magnitudes of~$O(1/\varepsilon\lambda)$, superconductivity is gradually reduced and eventually destroyed in increasing parts of the ring starting from the inner rim and moving outwards.   The analysis in Section~\ref{sec:weakSC} showed that surprisingly, 
the flux regime at which superconductivity diminishes is significantly smaller than the flux regimes studied in Section~\ref{sec:highCurrentRegime}.  This implies that for a range of flux values~$J$, the system~\eqref{eq:system_ring_region} has multiple non-trivial solutions.  It is not clear, however, from the asymptotic analysis, whether these solutions belong to different solution branches, bifurcate from a branch, or belong to one solution branch that undergoes branch folding.  In this section, we conduct a numerical continuation study in aim of mapping the solution space of system~\eqref{eq:system_ring_region}, and providing answers to the above questions.

\begin{figure}[ht!]
\begin{center}
\scalebox{0.5}{\includegraphics{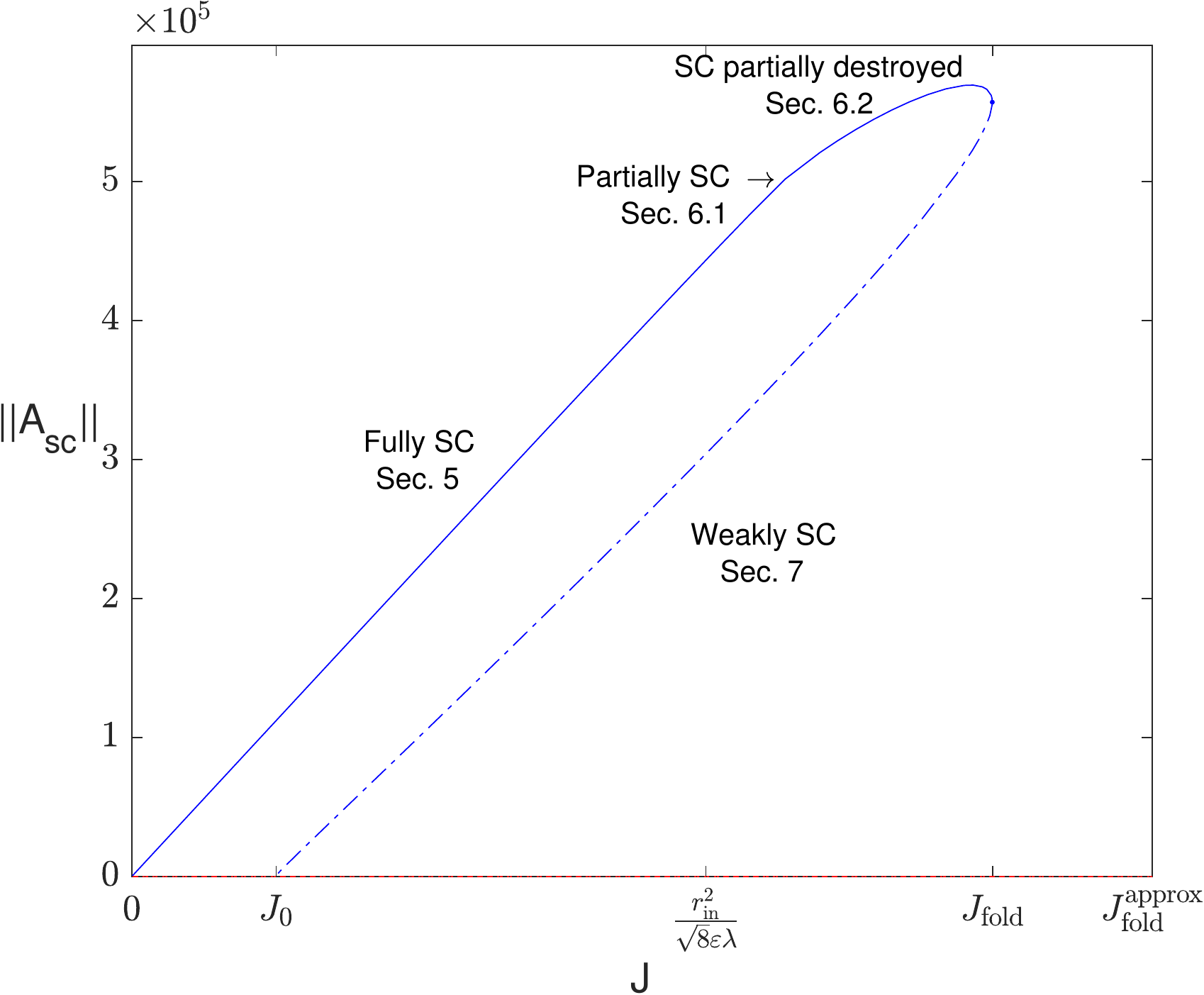}}
\caption{Bifurcation graph presenting the~$L_2$ norm~$\|A_{\rm sc}\|_2=\left[\int |A_{\rm sc}|^2d{\bf x}\right]^\frac12$ of~$A_{\rm sc}$ as a function of~$J$ where~$A_{\rm sc}$ is defined by~\eqref{eq:AJ} and~$A_J$ is the solution of~\eqref{eq:system_ring_region} with~$r_{\rm in}=1.5$,~$r_{\rm out}=2$,~$\varepsilon=0.001$ and~$\lambda=0.1$.  {\color{blue} Solid curve} is the branch of solutions characterized in Section~\ref{sec:superconductive} and~\ref{sec:highCurrentRegime} until the fold point~$J_{\rm fold}$.   {\color{blue} Dash-dotted curve} is the portion of the branch of solutions characterized in Section~\ref{sec:weakSC}.  {\color{red} Dotted curve} is the branch of trivial solutions~$A_{\rm sc}\equiv0$.  Note that the line style, solid or dash-dotted, does not provide an indication for the stability of the corresponding solutions.} 
\label{fig:bifurcationGraph_annotated}
\end{center}
\end{figure}
Figure~\ref{fig:bifurcationGraph_annotated} presents a bifurcation graph for system~\eqref{eq:system_ring_region} where~$J$ is the continuation parameter.   The trivial branch~$A_{\rm SC}\equiv0$ is the red dotted curve.  At~$J\ll1$, we observe that the system~\eqref{eq:system_ring_region} has a branch of non-trivial solutions marked by a solid blue curve.  These are the solution corresponding to a full superconductive state which were characterized in Section~\ref{sec:superconductive}, see, e.g., Figure~\ref{fig:generalPictureAJ_lowCurrent}.  When~$J>r_{\rm in}^2/\sqrt{8}\varepsilon\lambda$, the solutions along this branch, correspond to partial superconductive state, see, for example, Figure~\ref{fig:generalPicture_partialSC}.  In particular, the study of these solutions, see Section~\ref{sec:highCurrentRegime}, showed that, to leading order, they behave as solutions of system~\eqref{eq:system_ring_region} in the low flux case, but with an effective inner cylinder radius~$r_{\rm eff}$ which is larger than the inner cylinder radius~$r_{\rm in}$.   As discussed in Section~\ref{sec:highCurrentRegime}, the effective inner cylinder radius must reside in the cylinder region,~$r_{\rm in}\le r_{\rm eff}\le r_{\rm out}$.   This yields the bound
\[
J<\frac{r_{\rm out}^2}{\sqrt8\varepsilon\lambda},
\]
see~\eqref{eq:highCurrentRegime_destory_SC}, on the flux regime~$J$ in which such solution with partial superconductivity can exist.    
In terms of the bifurcation graph, we observe that the branch has a fold point at fluxes of the magnitude of this bound.  After the fold point, we observe a lower branch of solutions, marked by a blue dash-dotted curve.  As shown in Section~\ref{sec:weakSC}, these are the solutions corresponding to weak superconductivity, in which~$\psi(r_{\rm out})<1$, see, e.g., Figure~\ref{fig:generalPictureAJ_weakSC}.    This branch does not end at~$J=0$, but rather at~$J_0\approx r_{\rm out}/\varepsilon$, see~\eqref{eq:weakSC_currentregime}.

The location~$J_{\rm fold}$ of the fold point can be approximated by~\eqref{eq:highCurrentRegime_destory_SC}, i.e.,
\begin{equation}\label{eq:Jfold}
J_{\rm fold}\approx J_{\rm fold}^{\rm approx}=\frac{r_{\rm out}^2}{\sqrt8\varepsilon\lambda}.
\end{equation} 
We observe in Figure~\ref{fig:bifurcationGraph_annotated} that the approximation~\eqref{eq:Jfold} is not very accurate.  This is since Figure~\ref{fig:bifurcationGraph_annotated} corresponds to solutions of~\eqref{eq:system_ring_region} with a relatively large~$\lambda$.  
Figure~\ref{fig:Jfold} presents a graph of the fold point~$J_{\rm fold}$ computed numerically as a function of~$\lambda$ ({\color{blue} solid curve}), together with a graph of the approximated fold point~\eqref{eq:Jfold} ({\color{red} dashed curve}).   We observe that relative error reduces linearly with~$\lambda$.  Nevertheless, at~$\lambda=0.05$ the relative error is roughly 12\%.
\begin{figure}[ht!]
\begin{center}
\scalebox{0.66}{\includegraphics{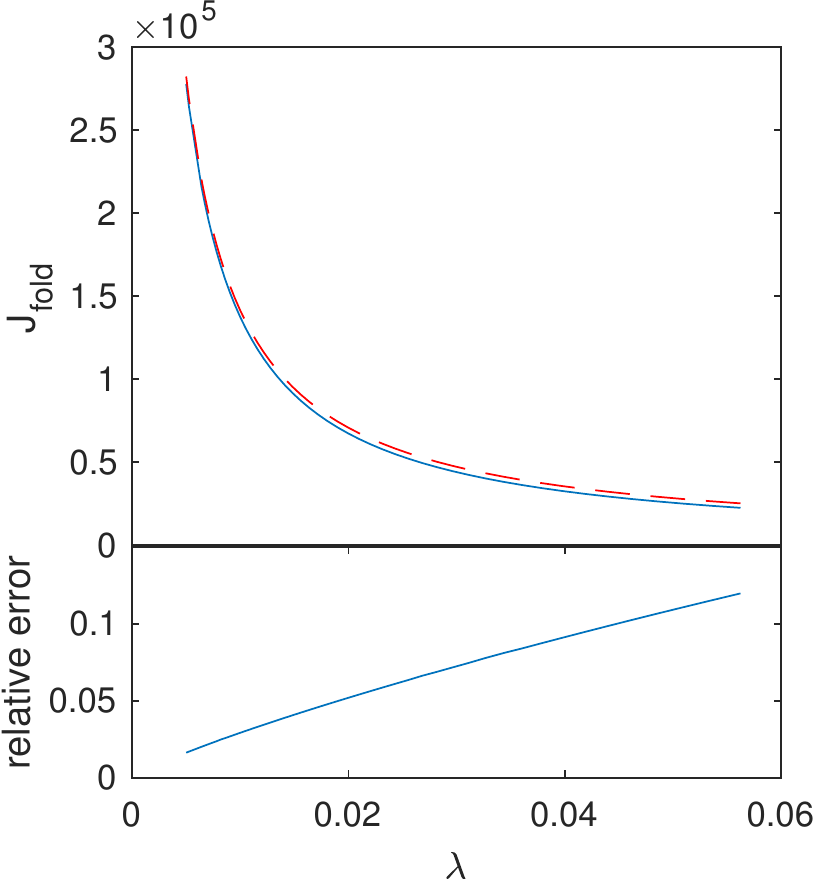}}
\caption{Top: Graph of the fold point~$J_{\rm fold}^{\rm numerical}$ as a function of~$\lambda$ ({\color{blue} solid}) as computed for~\eqref{eq:system_ring_region}
with~$r_{\rm in}=1.5$,~$r_{\rm out}=2$,~$\varepsilon=0.001$, and the approximated fold point~$J_{\rm fold}^{\rm approx}$ given by~\eqref{eq:Jfold} ({\color{red} dashes}).  Bottom graph is the relative error~$|J_{\rm fold}^{\rm numerical}-J_{\rm fold}^{\rm approx}|/J^{\rm numerical}_{\rm fold}$.}
\label{fig:Jfold}
\end{center}
\end{figure}

\begin{figure}[ht!]
\begin{center}
\scalebox{0.5}{\includegraphics{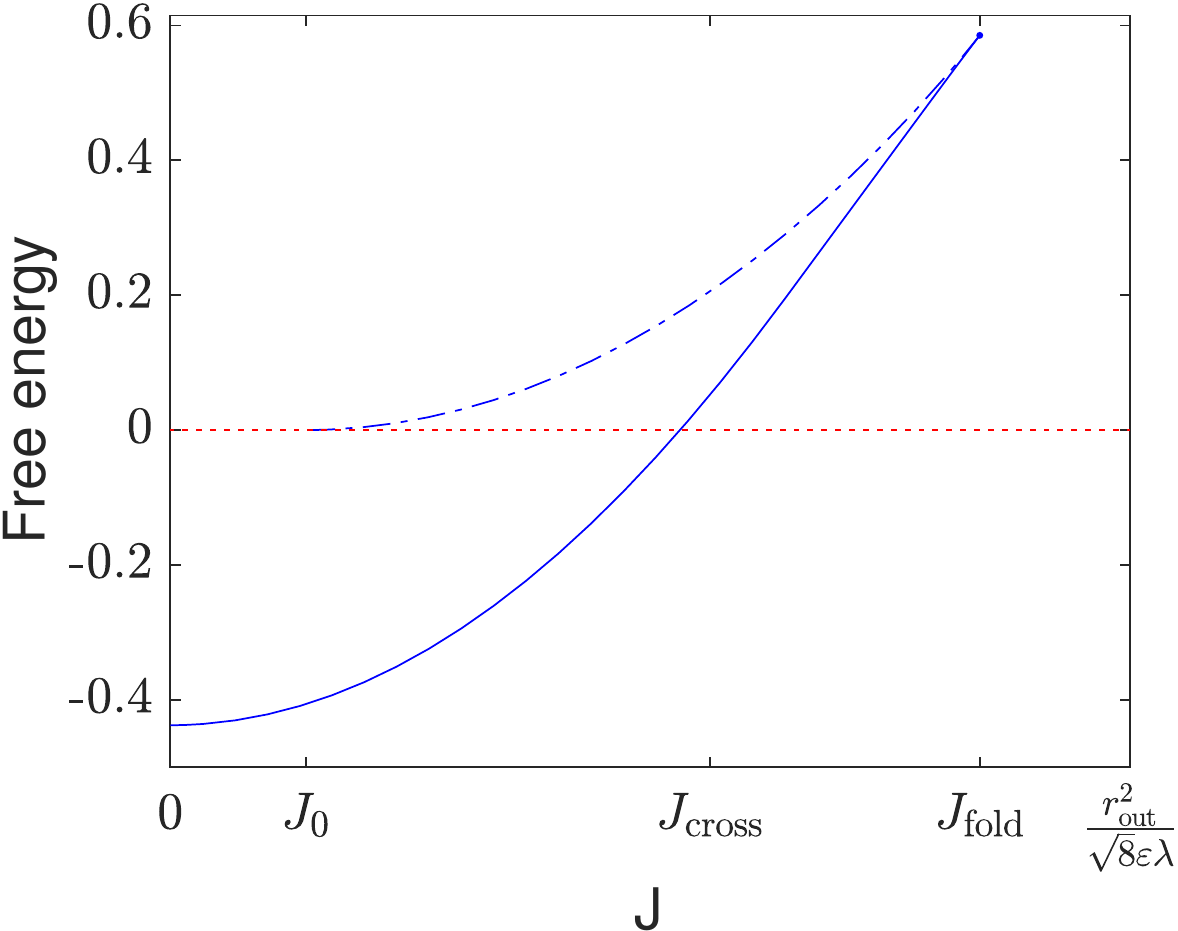}}
\caption{Bifurcation graph~$E(A_{\rm sc})$ as a function of~$J$ where~$\Energy$ is the free energy~\eqref{eq:energyAsc} and~$A_{\rm sc}$ is the 
 solution of~\eqref{eq:system_ring_region} with the same parameters as in Figure~\ref{fig:bifurcationGraph_annotated}:~$r_{\rm in}=1.5$,~$r_{\rm out}=2$,~$\varepsilon=0.001$ and~$\lambda=0.1$.  {\color{blue} Solid curve} is the branch of solutions characterized in Section~\ref{sec:superconductive} and~\ref{sec:highCurrentRegime} until the fold point~$J_{\rm fold}$ marked by the marker `{\color{blue} $\cdot$}'. {\color{blue} Dash-dotted curve} is the portion of the branch of solutions characterized in Section~\ref{sec:weakSC}.  {\color{red} Dotted curve} is the branch of trivial solutions~$A_{\rm sc}\equiv0$.  Note that the line style, solid or dash-dotted, does not relate to stability features of the corresponding solutions.} 
\label{fig:bifurcationGraph_energy}
\end{center}
\end{figure}The above analysis does not consider stability of the solutions, and accordingly the line style in Figure~\ref{fig:bifurcationGraph_annotated}, solid or dash-dotted, does not relate to stability features of the corresponding solutions.  Study of the related time-dependent system and the stability of the solutions of the steady-state system~\eqref{eq:system_ring_region} is beyond the scope of this work.  Often, however, physical systems are driven towards solution with lower free energy.  Figure~\ref{fig:bifurcationGraph_energy} presents the free energy~\eqref{eq:energyAsc} of the solutions of~\eqref{eq:system_ring_region} as a function of~$J$.  We observe that the branch of solutions corresponding to full or partial superconductivity, see solid curve, is energetically preferable over the branch of solutions corresponding to weak superconductivity, see dash-dotted curve.  We further observe that, for~$J\ll1$, the solutions corresponding to full superconductivity, see solid curve, have minimal free energy.   Yet for~$J>J_{\rm cross}$ where~$J_{\rm cross}=O(1/\varepsilon\lambda)$, the trivial solution (branch corresponding to {\color{red} dotted curve}) is energetically preferable.  

\begin{remark}
The analysis of Section~\ref{sec:superconductive} shows that for a narrow ring~$\sqrt{r_{\rm out}^2-r_{\rm in}^2}\ll\sqrt{r_{\rm out}^2+r_{\rm in}^2}$, the point at which~$\Energy(A_{\rm sc})=0$ resides in the fully conductive case satisfies
\[
J_{\rm cross}=\sqrt{r_{\rm out}^2-r_{\rm in}^2}\frac{r_{\rm in}}{\sqrt8\varepsilon\lambda}[1+O(\lambda)].
\]
\end{remark}

\section{Emerging picture from an experimentalist point of view}
The numerical continuation study together with the analysis conducted at different parameter regime give rise to a mapping of the solution space of system~\eqref{eq:system_ring_region}, as described in Section~\ref{sec:numericalContinuation}.  The motivation of this work is to study system~\eqref{eq:system_ring_region} to better understand the function of a stiffnessonometer device.  It is, now, instructive to revisit the emerging picture of the system from an experimentalist point of view.  

From an experimental point of view, there are two partially overlapping regimes of interest: (I) low flux $J \ll r^2_{\rm in}/\varepsilon\lambda$ (Sec.~\ref{sec:superconductive}), and (II) strong stiffness which usually arises at low temperatures where $\varepsilon \ll \lambda \ll 1$ (Sec.~\ref{sec:highCurrentRegime}, and \ref{sec:weakSC} ). In this work we have focused on regime II.  Nevertheless, our analysis is also valid in regime I.

In regime I, the vector potential at the pickup loop radius, $A_{\rm sc}(R_{\rm pl})$, in~\eqref{eq:Asc_J_dependence} could be used with~$A_{\rm coil}$ to extract $\lambda$. In regime II, we find three types of solutions, see Figure~\ref{fig:bifurcationGraph_annotated}:  The first type consists of solutions which are partially super-conductive (see Sections~\ref{sec:partialSC} and \ref{sec:SC_partially_destroyed}).  This brach of solution is continuously connected to $J=0$.  The second type consists of weakly super-conductive solutions, see Section~\ref{sec:weakSC}.  This branch of solution that does not connect directly to $J=0$, but rather through a folding point.  Finally, the third type is the trivial solution~$A_{\rm sc}\equiv0$.  

The partially super-conductive solutions are the most interesting type. As $J$ is ramped from zero, the solution is not different from regime I. But, for $J > {r_{\rm in}^2}/{\sqrt 8 \varepsilon \lambda }$, the order parameter's magnitude $\psi$ begins to diminish in the inner rim of the cylinder and the cylinder's hole in the regime of superconductivity effectively grows. Nevertheless, as long as $\psi=1$ over a region of length $\lambda$ somewhere in the SC, there is an outer region where the SC vector potential $A_{{\rm{sc}}}=-J/r$, which exactly cancels the applied vector potential $J/r$, giving a total $A=0$. Since $2\pi r A=\Phi $, it means that the flux generated by the SC exactly cancels the applied flux. The experiment is set to detect the SC flux, therefore, the signal will be linear in $J$ despite the destruction of SC in parts of the cylinder. With increasing $J$, the effective superconducting hole size increases until $\psi$ survives only on a boundary layer of width $\lambda$ at $r_{\rm out}$. This occurs at
\begin{equation}
J_{\rm fold} \lesssim \frac{{r_{\rm out}^2}}{{\sqrt 8 \varepsilon \lambda }}.
\end{equation}
The smaller $\varepsilon$ and $\lambda$ the better 
 the approximation is. At even larger $J$, the SC is no longer able to expel the applied flux and $A_{{\rm{sc}}}$ does no longer grow with $J$. A clear change of behavior in the $J$ dependence of $A_{{\rm{sc}}}$ is expected at $J_{\rm fold}$ allowing the determination of $\varepsilon$, given that $\lambda$ has been determined at lower flux values.

According to Figure~\ref{fig:bifurcationGraph_annotated}, the weakly super-conductive solutions have higher free energy than partially super-conductive solutions.  Since the entire exercise is based on finding the free energy minimum, such solutions are not expected to be observed experimentally. Moreover, since the branch of weakly super-conductive solution is not directly connected to $J=0$, there is no way to prepare such solutions even instantaneously. Applying current before cooling the SC is equivalent to setting the integer $m$ in Eq.~\eqref{eq:parameterslimits} such that $J-m$ is as close as possible to zero, which will send us back to the low $J$ solution.

As for the trivial solution, there is a crossing point at $J_{\rm cross}$ where the free energy of weakly super-conductive solutions is higher than that of the trivial solution. This might suggests that for $J>J_{\rm cross}$ the trivial solution is the relevant one. However, for this to occur, SC should disappear from a finite portion of the cylinder for an infinitesimal change in $J$, and the current should relax to zero. Nothing in the system can take this kinetic energy, and so it seems plausible that weakly superconductive solutions do not switch to the trivial solution upon increasing $J$ past $J_{\rm cross}$.

\section{Numerical details}\label{sec:numerics}
All simulations were conducted by pde2path~\cite{uecker2014pde2path} - a Matlab package for continuation and bifurcation in systems of PDEs.  
\section{Concluding remarks}\label{sec:ConcludingRemarks}
An interesting question is what happens when $J>J_{\rm fold}$. Our findings suggest that the system does not have a non-trivial solution for~$J>J_{\rm fold}$.  Hence, roughly speaking, in this case, the SC has ``no choice'' but to increase $m$ from zero so that $J-m \leq J_{\rm fold}$. In physical terms it means that a vortex is present at the center of the cylinder. This will lead to $A_{\rm{sc}}$ which is independent or decreases with increasing $J$. Vortex formation is beyond the scope of this work, and requires a study of the time dependent Ginzburg Landau equations.

Up to now, the Stiffnessometer has been used to collect data in regimes I and II.  This data was used to shed light on the nature of the phase transition in cuprates, and to show that upon cooling, SC first develops in two dimensions and only at lower temperature turn into a three dimensional phenomena~\cite{kapon2019phase}.  However,  due to the lack of theoretical understanding of the folding point, $\varepsilon$ has not been extracted, and analysis was restricted to parameter regimes far from the folding point. We anticipate that the derivation presented here will allow accurate determination of $\varepsilon$ in regimes not accessible before, and hopefully to new insights into the mechanisms of high temperature superconductivity.  Furthermore, the intuition obtained by this work could be tested experimentally by magnetic scanning techniques, and the concept that at high applied flux, the SC current is pushed to the outer radius should be examined.

The analysis we presented used the explicit standard form
$V=\alpha \psi^2+\beta \psi^4$ of the Ginzburg-Landau potential.  The question of what happens when higher order approximations of the Ginzburg-Landau potential are taken into account is open.
It is common wisdom that the main properties of superconductors are
controlled by the behavior of $V(\psi)$ near its extremal points $\psi=0,1$.
This is due to the fact that outside a small boundary layer, the value of $\psi$
is always close to one of these extrema.
One therefore expects that modifying $V(\psi)$, say by adding a $\psi^6$ term,
would only change the profile of $\psi(r)$ within the boundary layer without significantly
affecting the outer solution.

From a mathematical point of view, the study of the underlying system gives rise to demanding nonlinear turning point problems, in which the location of the turning point is a-priori unknown and the choice of the matching direction is non-trivial.  In this work, we avoided in a sense these problems using variational approximations.  Asymptotic analysis of the underlying nonlinear turning point problem is left for future study.  

Finally, in this work the super-conducting region we have considered was a hollow cylinder of infinite height. This choice allowed us to obtain the ODE system~\eqref{eq:system_ring_region} for the quantities of study.  In an actual system, the super-conducting region is a ring of finite height of magnitude comparable to its radius or much less.   In this case, the quantities of study are described by a PDE system.  Study of stiffnessonometer systems with super-conducting rings of finite (or zero) height will be presented elsewhere.
\subsection*{Acknowledgments}
This study was financially supported by Israeli Science Foundation (ISF) grant number 315/17.

\appendix

\end{document}